\documentclass[journal,twoside,web]{ieeecolor}

\usepackage{generic}
\usepackage{cite}
\usepackage{amsmath,amssymb,amsfonts}
\usepackage{bbm}

\usepackage{amsthm}

\usepackage{graphicx}
\usepackage{hyperref}
\usepackage{textcomp}
\usepackage[acronym]{glossaries}
\usepackage{dsfont}

\usepackage{rotating}
\usepackage{booktabs}

\usepackage{textcomp}
\usepackage{listings}
\usepackage{colortbl}
\usepackage{etoolbox}
\usepackage{graphicx}
\usepackage{xspace}
\usepackage{url}
\usepackage{booktabs}
\usepackage{subcaption}
\usepackage{cite}

\usepackage[utf8]{inputenc}
\usepackage{multirow}
\usepackage[T1]{fontenc}
    
\def\BibTeX{{\rm B\kern-.05em{\sc i\kern-.025em b}\kern-.08em
    T\kern-.1667em\lower.7ex\hbox{E}\kern-.125emX}}

\newcommand{\unsafe}{\mathrm{u}}
\newcommand{\safe}{\mathrm{s}}
\newcommand{\reach}{\mathrm{r}}
\newcommand{\initial}{\mathrm{0}}
\newcommand{\set}[1]{\{#1\}}

\newcommand{\algebraic}{X = \big \{ x \in \reals^n \mid h_i(x) \geq 0 \; \forall i \in \set{1, \ldots, l} \big \}}

\newcommand{\pw}{\mathbf{w}}
\newcommand{\px}{\mathbf{x}}

\newcommand{\pX}{\mathbf{x}}

\newcommand{\nx}{\mathbf{x}'}

\newcommand{\reals}{\mathbb{R}}
\newcommand{\naturals}{\mathbb{N}}
\newcommand{\N}{\mathcal{N}}
\newcommand{\hor}{H}
\newcommand{\RA}{\mathcal{R}}
\newcommand{\naturalszero}{\mathbb{N}_{\geq 0}}

\renewcommand{\L}{\mathcal{L}}

\newcommand{\low}[1]{\underline{#1}}

\newcommand{\Pfin}{P_{\reach}(X_{\reach}, X_u, X_0, \hor, \pi)}
\newcommand{\Pfinopt}{P_{\reach}(X_{\reach}, X_u, X_0, \hor, \pi^{\ast})}
\newcommand{\Pinf}{P_{\reach}(X_{\reach}, X_u, X_0, \pi)}

\theoremstyle{definition}
\newtheorem{definition}{Definition}
\newtheorem{problem}{Problem}

\theoremstyle{plain}
\newtheorem{theorem}{Theorem}
\newtheorem{corollary}[theorem]{Corollary}
\newtheorem{proposition}[theorem]{Proposition}
\newtheorem{lemma}[theorem]{Lemma}

\theoremstyle{plain}
\newtheorem{remark}{Remark}

\def\BibTeX{{\rm B\kern-.05em{\sc i\kern-.025em b}\kern-.08em
    T\kern-.1667em\lower.7ex\hbox{E}\kern-.125emX}}
\begin{document}
\title{Time-Varying Reach-Avoid Control Certificates for Stochastic Systems}
\author{Rayan Mazouz, Luca Laurenti, Morteza Lahijanian
\thanks{Rayan Mazouz is with the University of Colorado Boulder (e-mail: \href{mailto:rayan.mazouz@colorado.edu}{rayan.mazouz@colorado.edu}).}
\thanks{Luca Laurenti is with Delft University of Technology (e-mail:
\href{mailto:l.laurenti@tudelft.nl}{l.laurenti@tudelft.nl)}).}
\thanks{Morteza Lahijanian is with the University of Colorado Boulder (e-mail: 
\href{mailto:morteza.lahijanian@colorado.edu}{morteza.lahijanian@colorado.edu)}.}}

\maketitle

\begin{abstract}
Reach-avoid analysis is fundamental to reasoning about the safety and goal-reaching behavior of dynamical systems, and serves as a foundation for specifying and verifying more complex control objectives. This paper introduces a reach-avoid certificate framework for discrete-time, continuous-space stochastic systems over both finite- and infinite-horizon settings. We propose two formulations: time-varying and time-invariant certificates. We also show how these certificates can be synthesized using sum-of-squares (SOS) optimization, providing a convex formulation for verifying a given controller. Furthermore, we present an SOS-based method for the joint synthesis of an optimal feedback controller and its corresponding reach-avoid certificate, enabling the maximization of the probability of reaching the target set while avoiding unsafe regions. Case studies and benchmark results demonstrate the efficacy of the proposed framework in certifying and controlling stochastic systems with continuous state and action spaces.
\end{abstract}

\begin{IEEEkeywords}
Reach-Avoid Certificates, Stochastic Systems, Probabilistic Safety, Convex Optimization, Formal Verification
\end{IEEEkeywords}

\section{Introduction}


Ensuring the safe operation of autonomous systems in \emph{safety-critical} domains is a central objective in modern control and robotics \cite{lesser2013stochastic, li2021prediction}. Real-world systems operate under uncertainty and stochastic disturbances, making their behavior difficult to predict. It is therefore necessary to rigorously verify these systems or synthesize controllers that provide \emph{formal} probabilistic guarantees on their behavior. This task is particularly challenging due to the continuous nature of the action space, the evolution of the system over a continuous state space under possibly nonlinear dynamics and stochastic uncertainty, and the complexity of the specification requirements. A fundamental requirement in many applications is the \emph{reach–avoid} property, which simultaneously captures safety constraints and goal-reaching objectives. In this work, we address these challenges by developing a certificate-based framework that provides formal reach–avoid guarantees while enabling efficient synthesis of feedback controllers.

Existing methods for reach–avoid analysis of continuous-space, discrete-time stochastic systems generally fall into two categories. The first class consists of finite abstraction-based approaches \cite{lahijanian2015formal, gracia2025efficient, lunze2003discrete, coppola2024data, kariotoglou2017linear}, which approximate the system by an (uncertain) finite-state Markov process. These methods require discretizing the state and action spaces, which becomes computationally prohibitive in high dimensions. Moreover, they rely on computing bounds on the transition kernel between abstract states, which is especially challenging for nonlinear dynamics.
The second class comprises continuous certificates subject to reach–avoid constraints \cite{vzikelic2023learning, wicker2024probabilistic}. Obtaining such certificates requires solving a functional optimization problem under these constraints. However, existing formulations are limited to infinite-horizon settings and 
become conservative when cast as convex optimization problems.
The control synthesis component further exacerbates this challenge by adding additional non-convexity.
To circumvent this issue, existing works often rely on deep learning to approximate the certificate and controller, but such approaches remain difficult to verify against hard constraints.

In this paper, we focus on reach–avoid certificates for discrete-time nonlinear stochastic systems with continuous state and action spaces. We introduce a formulation of such certificates that is grounded in Bellman's optimality principle, i.e., dynamic programming (DP),  and bridges the gap between classical value-iteration-based methods and analysis over continuous domains. Our framework applies to both \emph{finite-} and \emph{infinite-horizon} settings and accommodates both \emph{time-invariant} and \emph{time-varying} certificates.
Moreover, in contrast to \cite{vzikelic2023learning, wicker2024probabilistic}, our formulation naturally admits a convex optimization structure without introducing additional conservatism: by restricting the certificate to polynomial functions, we can leverage semidefinite programming (SDP) techniques~\cite{vandenberghe1996semidefinite}. In particular, we employ Sum-of-Squares (SOS) optimization for synthesis of both certificate and controller. 
%
%
We evaluate and compare these certificate formulations and synthesis frameworks on several linear and nonlinear systems, confirming the theoretical guarantees and illustrating their characteristics and trade-offs.

In short, the main contributions are:
\begin{itemize}
    \item A unified formulation of reach–avoid certificates, both \emph{time-varying} and \emph{time-invariant}, for discrete-time, continuous-space stochastic systems, applicable to both \emph{finite}- and \emph{infinite-horizon} settings, with guaranteed lower bounds on the probability of success.
     \item A convex optimization framework for joint synthesis of the reach-avoid certificate and feedback controller via SOS polynomials.
    \item Validation of the theoretical guarantees,
    along with benchmarking comparisons that highlight the trade-offs among the proposed certificate formulations and their advantages over state-of-the-art formulations.
\end{itemize}

\subsection*{Related Work}
Ensuring \emph{reach-avoid} property in stochastic systems requires methods that remain robust under uncertainty. Many existing approaches focus solely on the analysis of safety. Stochastic barrier functions are especially powerful in this regard, as they can reduce safety verification and control design to a convex optimization problem. For verification, these barrier certificates provide formal probabilistic safety guarantees for systems with process noise \cite{SANTOYO2021109439, mazouz2022safety, mazouz2024piecewise}, while control barrier functions enable safe controller synthesis for stochastic processes \cite{mazouz2024data, jagtap2020formal,  vahs2023belief, mazouz2025piecewise}. However, distinct from this work, these methods do not explicitly incorporate a \emph{reach} objective into their formulation and therefore do not address the full reach–avoid properties common in real-world applications.

In fact, reach–avoid properties generalize both reachability and safety objectives \cite{summers2010verification} and are foundational for complex requirements such as the ones expressed in temporal logics.
Much of the existing literature approach this problem through finite abstraction of the continuous-space, continuous-action  system via discretization. In \cite{lahijanian2015formal},  interval Markov decision process (IMDP) abstractions of such stochastic systems are introduced, along with methods for verification, controller synthesis, and refinement. In \cite{gracia2025efficient}, an approach is developed for synthesizing distributionally robust reach-avoid controllers for switched stochastic systems. Work \cite{kariotoglou2017linear} formulates the stochastic reach-avoid control problem via an LP with finite-dimensional approximations. In \cite{coppola2024data}, data-driven abstractions of stochastic systems are used to synthesize controllers that satisfy reach-avoid specifications.
In this work, we focus on an approach that does not require space or action discretization.

Approaches that operate directly in continuous spaces use certificate formulations for discrete-time~\cite{vzikelic2023learning,xue2025finite} and continuous-time~\cite{neustroev2025neural} stochastic systems. 
Closest to our work are~\cite{xue2025finite,vzikelic2023learning}, which derive time-invariant reach-avoid certificates for finite~\cite{xue2025finite} and infinite horizons~\cite{vzikelic2023learning}.
In~\cite{xue2025finite}, barrier-like conditions yield upper and lower bounds on finite-horizon reach-avoid probabilities. However, these bounds tend to be conservative and the resulting non-convex certificate synthesis limits scalability to low-dimensional systems. 
%
The work in~\cite{vzikelic2023learning} uses supermartingale theory to derive the certificate and proposes a joint learning framework for neural network certificates and control policies via alternating training and verification. However, it is limited by (i) the time-invariant structure, which requires highly expressive (large) neural networks, and (ii) the scalability of neural network verification methods, which are restricted to simple architectures.

To address limitation (i), we develop conditions for both time-varying and time-invariant certificates for finite- and infinite-horizon properties, derived from a DP formulation, which lead to tighter bounds on the reach–avoid probability bound. We show that the certificate in~\cite{vzikelic2023learning} is a special case of our formulation 
and our probability bounds are tighter than the one in \cite{xue2025finite}.
To address limitation (ii), we focus on convex formulations of the functional optimization problem for the joint synthesis of the certificate and controller, avoiding the need for NN verification.  Specifically, we focus on SOS programs, and show that our certificate formulations are particularly suitable for such templates.



Finally, we note that there are other methods that leverage non-convex optimization techniques to synthesize certificates, specifically satisfiability modulo theory (SMT) based approaches for reach-avoid verification and controller synthesis \cite{fan2021controller,mufid2021smt,cao2025comparative}. In \cite{fan2021controller}, a sound control synthesis algorithm is presented for deterministic linear systems with reach-avoid specifications, formulated as a satisfiability problem. Similarly, \cite{mufid2021smt} performs reachability analysis on deterministic high-dimensional linear systems using SMT-based reasoning. A comparative analysis between barrier-like functions methods for stochastic reach-avoid verification is provided in \cite{cao2025comparative}, evaluating the practical performance of SDP- and SMT-based approaches. However, none of these methods explicitly incorporate time-varying conditions that capture the temporal structure inherent in reach-avoid problems, and certificate and controller synthesis are performed separately rather than jointly optimized.
Moreover, while ensuring numerical soundness, SMT-based methods suffer from poor scalability and lack of completeness since constraint solving over nonlinear real arithmetic is undecidable \cite{gao2013dreal}.

In contrast, in this work we focus on convex optimization-based formulations to ensure tractability and to enable the simultaneous synthesis of both the certificate and the controller within a unified SOS-based convex program. Importantly, concerns regarding the numerical stability and validity of SOS optimization tools have been addressed in recent years through the development of more numerically stable SOS solvers \cite{legat2017sos, mazouz2026stochasticbarrier} and methods for certifying SOS programs \cite{wajid2025successive, panja2025correct}. As a result, recent stochastic barrier approaches \cite{mazouz2026stochasticbarrier} demonstrate significantly improved reliability and performance compared to earlier methods. These advances further motivate our choice of SOS program as a principled and computationally tractable framework for synthesizing certificates. Nevertheless, we emphasize that our certificate formulations are independent of the particular optimization method used.

\section{Problem Formulation}

We consider a discrete-time, continuous-space stochastic system given by
\begin{equation}
\label{eq:system}
    \px_{k+1} = f(\px_k, u_k, \pw_k),
\end{equation}
where $\px_k \in X \subseteq \mathbb{R}^n$ denotes the state and $u_k \in U \subset \mathbb{R}^m$ is the control input at time step $k \in \naturalszero$. The disturbance $\pw_k\in W \subseteq \mathbb{R}^w$ is an i.i.d. random variable with distribution $p_{\mathbf{w}}$, assumed to admit a closed-form moment generating function (e.g., $\pw_k \sim \mathcal{N}(0,\Sigma)$). 
The system dynamics are governed by a polynomial map $f: X \times U \times W \to \mathbb{R}^n$ which is continuous in its arguments. 

The choice of control $u_k$ is determined by a time-dependent \emph{feedback controller} $\pi: X \times \naturals \to U$.
Given a state $x \in \mathbb{R}^n$, a controller $\pi$, and a Borel-measurable set $X' \subseteq X$, we define the transition kernel of System~\eqref{eq:system} as
\begin{equation}
\label{eq:transition_kernel}
    T(X' \mid x,\pi) := \int_{\mathbb{R}^{\mathrm{w}}} \mathbf{1}_{X'}(f(x,\pi(x), \pw ) \, p_{\mathbf{w}}(w) \, dw,
\end{equation}
where $\mathbf{1}_{X'}(\cdot)$ denotes the indicator function of the set $X'$, taking the value 1 when its argument belongs to $X'$, and 0 otherwise.

Under a controller $\pi$, from an initial state $x_0 \in X$, this transition kernel induces a well-defined probability measure $\mathrm{Pr}^{x_0}$ over the trajectories of System~\eqref{eq:system}~\cite{bertsekas2004stochastic}. For any time-step $k \in \mathbb{N}_0$, initial state $x_0 \in \mathbb{R}^n$, and measurable sets $X_0, X_k \subseteq X$, the measure satisfies
\begin{subequations}
\label{eq:prob measure}
    \begin{align}
        &\mathrm{Pr}^{x_0}[\pX_0 \in X_0] = \mathbf{1}_{X_0}(x_0), \\
        &\mathrm{Pr}^{x_0}[\pX_k \in X_k \mid \pX_{k-1} = x, \, \pi] = T(X_k \mid x, \pi(x)).
    \end{align}
\end{subequations}
The probability measure $\mathrm{Pr}$ enables probabilistic reasoning over the stochastic trajectories of System~\eqref{eq:system}. 

%

We are interested in the temporal properties of System~\eqref{eq:system}. In this work, we specifically focus on \textit{reach-avoid} properties, which serve as a fundamental building block for more complex specifications, such as those expressed using temporal logics (e.g., LTL and its variants).  

\begin{definition}[Probabilistic Reach-Avoid]
    \label{def: probabilistic reach-avoid}
    Let $X_\safe \subset X$ represent the safe set, $X_0 \subseteq X_\safe$ the initial set, and $X_\reach \subseteq X_\safe$ the desirable goal (reach) set. Define unsafe set by $X_\unsafe := X \setminus X_\safe$. 
    Then, given controller $\pi$ and horizon $\hor \in \mathbb{N}_0 \cup \{\infty\}$, \emph{probabilistic reach-avoid} is defined as
    \begin{multline*}
        \Pfin = \inf_{x_0 \in X_0} \mathrm{Pr}^{x_0}[ (\exists k \leq \hor, \ \pX_k \in X_\reach) \\
        \wedge 
        (\forall k' < k, \ \pX_{k'} \not\in X_\unsafe)]. 
    \end{multline*}
    When $\hor = \infty$, we use the simplified notation $\Pinf. $
\end{definition}

In this work, we consider the following verification problem.

\begin{problem}[Reach-Avoid Verification]
\label{Prob:ReachAvoidCertification}
Consider System~\eqref{eq:system} with semi-algebraic control set $U$ and 
let $X_\safe \subset \reals^n$,
$X_0 \subseteq X_\safe$,  and
$X_\reach \subseteq X_\safe$ be semi-algebraic sets that define the safe, initial, and target regions per Definition~\ref{def: probabilistic reach-avoid}.
For a given feedback controller $\pi$, horizon $\hor \in \mathbb{N}_0 \cup \{\infty\}$, and probability threshold $\delta \in [0,1]$,
verify that
    \begin{align*}
        \Pfin \geq \delta.
    \end{align*}
\end{problem}


We also consider the problem of generating a controller that guarantees a lower-bound on the reach-avoid probability.

\begin{problem}[Reach-Avoid Control Synthesis]
\label{Prob:ReachAvoidSynthesis}
Consider the setting in Problem~\ref{Prob:ReachAvoidCertification}.  Synthesize a feedback controller $\pi^*$ for System~\eqref{eq:system} such that
    \begin{align*}
    \Pfinopt \geq \delta.
    \end{align*}
\end{problem}

The above problems are challenging because the system is stochastic and evolves in a continuous state space under nonlinear (polynomial) dynamics $f$. Moreover, the reach–avoid objective requires reasoning about two distinct sets, $X_u$ and $X_r$, simultaneously, which complicates the analysis. These factors together make the exact computation of reach–avoid probabilities nontrivial. In addition, synthesizing a feedback controller over a continuous control space $U \subset \reals^m$ further compounds the computational difficulty.

To obtain a tractable approach, we first derive a form of certificate that is amenable to computation. Specifically, we introduce \emph{time-varying} and \emph{time-invariant} reach–avoid certificates, defined as real-valued functions that satisfy a set of sufficient conditions. We then show that any such certificate yields a lower bound on $P_\reach$, thereby reformulating Problem~\ref{Prob:ReachAvoidCertification} as a functional optimization problem that seeks a certificate maximizing this bound. Likewise, Problem~\ref{Prob:ReachAvoidSynthesis} becomes a joint optimization problem over both the certificate and the controller. By restricting the certificate to the class of polynomials, we show that both problems admit convex relaxations in the form of semidefinite programs using SOS.

\begin{remark}
    We assume polynomial dynamics $f$ in System~\eqref{eq:system} for clarity of presentation; however, the proposed approach can be extended to general nonlinear dynamics by, e.g.,  constructing (piecewise) convex relaxations of $f$.
\end{remark}
\section{Preliminaries: Non-negative Polynomials}
\label{sec:SOS}

In this section, we provide a brief background on non-negative and SOS polynomials, which are central components of our approach.



\begin{definition}[SOS Polynomial]
\label{def:sosp1}
A multivariate polynomial $\lambda(x)$ is a Sum-Of-Squares (SOS) for $x \in \mathbb{R}^{n}$ if there exist polynomials $\lambda_{i}$, $i = 1, \ldots, r$, for some ${r \in \naturals}$, such that
$\lambda(x) = \sum_{i=1}^{r} \lambda_{i}^{2}(x).$
If $\lambda(x)$ is an SOS, then $\lambda(x) \geq 0$ for all $x \in \mathbb{R}^{n}$. The set of SOS polynomials is denoted by $\Lambda$.
\end{definition}

Now, consider a basic closed semi-algebraic set $\algebraic$, defined as the intersection of $l$ polynomial inequalities. Proposition \ref{prop:sosp1} provides a method to enforce non-negativity on such sets.
\begin{proposition}[\cite{Stengle1974}]
\label{prop:sosp1}
Let $X \subset \reals^n$ be a compact basic semi-algebraic set defined as
$\algebraic,$
where $h_i(x)$ is a polynomial. Then, polynomial $\gamma(x)$ is non-negative on $X$ if, for some $\lambda_{0}(x), \lambda_{i}(x) \in \Lambda$,
$\gamma (x) = \lambda_0(x) + \sum_{i=1}^{l} \lambda_{i}(x) h_{i}(x).$
\end{proposition}

\begin{corollary}
\label{cor:sosp1}
Let $\gamma(x)$ be a non-negative polynomial on $\algebraic,$ and $h(x)$ be an $l$-dimensional vector of polynomials with $h_i(x)$ as its $i$-th component.
Further, let $\L$ be an $l$-dimensional vector of SOS polynomials, i.e., the $i$-th element of $\L$ is $\lambda_i(x) \in \Lambda$.
Then, it holds that $ \gamma(x) - \mathcal{L}^{T} h(x) \in \Lambda$.
\end{corollary}

\section{Reach-Avoid Certificates}
\label{sec: certificate formulations}


In this section, we present general conditions for stochastic reach–avoid certificates without assuming algebraic sets or polynomial dynamics. 

To formally reason about reach–avoid properties, we seek functions that certify a lower bound on the probability that trajectories starting in the initial set $X_0$ reach the target set $X_r$, while remaining within the safe set $X_s$ (i.e., without entering the unsafe set $X_u$). 
Inspired by the Dynamic Programming (DP) perspective on stochastic barrier functions introduced in~\cite{laurenti2025unifying}, 
we derive two classes of reach-avoid certificates: \emph{time-varying} and \emph{time-invariant}, each applicable to both finite- and infinite-horizon settings.
Theorem~\ref{th:reach-avoid value function} establishes the formulation of the \emph{time-varying} certificate.


\begin{theorem}[Time-Varying, Bounded-Horizon Reach-Avoid Certificate]
\label{th:reach-avoid value function}
    Consider System~\eqref{eq:system} 
    under feedback controller $\pi$, 
    with a reach-avoid property defined by safe set $X_s$, initial set $X_0$, target set $X_\reach$, unsafe set $X_u$, and horizon $\hor \in \naturals_0$ per Definition~\ref{def: probabilistic reach-avoid}. Let $\nx = f(x, \pi, \textbf{w})$. Then, the real-valued function $\RA : \mathbb{R}^n \times \naturals^{\leq \hor}_{\geq 0} \to \mathbb{R}_{\geq 0}$ that is continuous in its first argument is a \emph{reach-avoid} certificate if, for some constants $\alpha_0, \alpha_i,\beta_i,\gamma \in \reals_{\geq 0}$ with $i \in \naturals^{\leq \hor}_{\geq 1}$, it satisfies 
    \begin{subequations}
        \label{eq: certificate tv finite H}
        \begin{align}
            \!\!\!\!\!&\RA(x, i) \geq 0 && \forall x \in \mathbb{R}^n, \ \forall i \in \naturals^{\leq \hor}_{\geq 0} \label{eq:nonneg-tv} \\
            &\RA(x, i) \leq 1 + \alpha_i && \forall x \in X_\reach, \ \forall i \in \naturals^{\leq \hor}_{\geq 0} \label{eq:reach-tv} \\
            &\RA(x, i) \leq \alpha_i && \forall x \in X_\unsafe, \ \forall i \in \naturals^{\leq \hor}_{\geq 0} \label{eq:nsafe-tv} \\
            &\RA(x, 0) \leq \alpha_0 && \forall x \in X_\safe\setminus X_\reach,  \label{eq:initial-nsafe-tv} \\
            &\RA(x, i) \leq \mathbb{E}[\RA(\px', i-1) | x ]  && \nonumber \\
            &  \qquad \qquad  - \alpha_{i-1} + \alpha_i + \beta_i 
                 && \forall x \in X_\safe \! \setminus \! X_\reach,  \forall i \in \naturals^{\leq \hor}_{\geq 1} \label{eq:exp-tv}\\
            &\RA(x, \hor) \geq \gamma && \forall x \in X_0. \label{eq:Vinitial_ss}
        \end{align}
    \end{subequations}
    Then, the reach-avoid probability is lower-bounded by
    \begin{align}
        \label{eq: probability bound for TV reach-avoid}
        \Pfin \geq \gamma - \left( \alpha_\hor + \sum_{i=1}^{ \hor} \beta_i
        \right ).
    \end{align}  
\end{theorem}
\begin{proof}
    The proof follows from the DP formulation for reach-avoid probability.  Let 
    $V_k(x) = P_\reach(X_r, X_u, \{x\}, X_0, k, \pi)$ be the probability of satisfying the reach-avoid property from $x \in \reals^n$ in $k$ time steps \cite{wicker2024probabilistic}.  Then, $V_k(x)$ is given by the following DP
    \begin{align}
        \label{eq: value function}
        V_{k}(x) = 
        \begin{cases}
           1 & \text{if } x \in X_r\\
           0 & \text{if } x \in X_u\\
           0 & \text{if } x \in (X_s \setminus X_r) \land k = 0\\
           \mathbb{E}[V_{k-1}(x') \mid x] & \text{if } x \in (X_s \setminus X_r) \land k > 0\\
        \end{cases}
    \end{align}
    Next, we show by induction that $V_k (x) \geq \RA(x,k) - \alpha_k - \sum_{i=1}^k \beta_i$ for all $k \geq 0$.
    First note that, for $k = 0$, $V_0(x) \geq \RA(x,0) - \alpha_0$ for all $x \in \reals^n$.  
    Next, assume that
    \begin{align*}
        V_{k}(x) \geq 
            \begin{cases}
               \RA(x,k) - \alpha_k & \text{if } x \in X_r\\
               \RA(x,k) - \alpha_k & \text{if } x \in X_u\\
               \RA(x,k) - \alpha_k - \sum_{i=1}^{k} \beta_i & \text{if } x \in (X_s \setminus X_r).
            \end{cases}
    \end{align*}
    Then, it follows that 
    \begin{align*}
    V_{k+1}(x) \geq 
        \begin{cases}
           \RA(x,k+1) - \alpha_{k+1} \hspace{25mm}  \text{if } x \in X_r\\
           \RA(x,k+1) - \alpha_{k+1} \hspace{25mm}  \text{if } x \in X_u\\
           \mathbb{E}[\RA(x',k) \mid x] - \alpha_k - \sum_{i=1}^{k} \beta_i  
            \\
           \quad \geq \RA(x,k+1) - \alpha_{k+1} - \sum_{i=1}^{k+1} \beta_i  \\\hspace{44mm} \text{if } x \in (X_s \setminus X_r)
        \end{cases}
    \end{align*}    
    Hence, it holds that
    \begin{align*}
       \Pfin & \geq \inf_{x \in X_0} V_H(x) \\
        & \geq \inf_{x \in X_0} \left(\RA(x,\hor) - \alpha_\hor - \sum_{i=1}^{\hor} \beta_i \right) \\
        & = \gamma - \left(\alpha_\hor + \sum_{i=1}^{\hor} \beta_i \right)
    \end{align*}
\end{proof}

From the DP perspective, the reach–avoid probability of System~\eqref{eq:system}, from a state $x \in X$ at time $k$, is given by the value function $V_k(x)$ in \eqref{eq: value function}. Intuitively, Theorem~\ref{th:reach-avoid value function} provides conditions under which a certificate $\RA$ yields a lower bound on $V_k$. However, note that $V_0$ is an indicator function at $k=0$, which is difficult to approximate with a continuous function. To address this, 
in constraints~\eqref{eq: certificate tv finite H}
we introduce a slack variable $\alpha_0 > 0$ that allows for a conservative relaxation of the indicator.

We extend this relaxation to every time step, by introducing slack variables $\alpha_k > 0$ to conservatively relax the hard constraints that $V_k$ takes values $1$ on $X_r$ and $0$ on $X_u$. Additionally, we introduce $\beta_i > 0$ in constraint \eqref{eq:exp-tv} to relax the requirement that $V_k(x) = \mathbb{E}[V_{k-1}(x') \mid x]$ on $X_s \setminus X_r,$
a condition that is otherwise difficult to enforce directly through optimization.
Putting these together, the conditions in \eqref{eq: certificate tv finite H} guarantee that $V_k(x) \ge \RA(x,k) - \alpha_k$ on $X_r$ and $X_u$, and
$V_k(x) \ge \RA(x,k) - \alpha_k - \sum_{i=1}^k \beta_i$ on $X \setminus X_r$.
In other words, these bounds ensure that the certificate $\RA$ provides a valid, conservative lower bound on the reach–avoid value function.


Based on this intuition, we can then easily extend the 
the bounded-horizon certificate in Theorem~\ref{th:reach-avoid value function} to unbounded-horizon by ensuring that $\beta_k+\alpha_k = 0$ for some $k < \infty$.

\begin{corollary}[Time-Varying, Unbounded-Horizon Reach-Avoid Certificate]
\label{corollary:reach-avoid infinite horizon}
    Consider the setting in Theorem~\ref{th:reach-avoid value function} with an unbounded-horizon reach-avoid property: $\hor = \infty$.
    Then, the real-valued function $\RA : \mathbb{R}^n \times \naturals^{\leq k+1}_{\geq 0} \to \mathbb{R}_{\geq 0}$ for some $k < \infty$ that is continuous in its first argument is a \emph{reach-avoid} certificate if, for some constants $\alpha_0, \alpha_i, \alpha_{k+1}, \beta_i,\gamma \in \reals_{\geq 0}$ with $i \in \naturals^{\leq k}_{\geq 1}$, it satisfies 
    \begin{subequations}
        \begin{align}
            \!\!\!\!\!&\RA(x, i) \geq 0 && \hspace{-2mm} \forall x \in \mathbb{R}^n, \ \forall i \in \naturals^{\leq k+1}_{\geq 0} \label{tv-bounded-1} \\
            &\RA(x, i) \leq 1 + \alpha_i && \hspace{-2mm} \forall x \in X_\reach, \ \forall i \in \naturals^{\leq k+1}_{\geq 0}  \label{tv-bounded-2} \\
            &\RA(x, i) \leq \alpha_i && \hspace{-2mm} \forall x \in X_\unsafe, \ \forall i \in \naturals^{\leq k+1}_{\geq 0} \label{tv-bounded-3}  \\ 
            &\RA(x, 0) \leq \alpha_0 && \hspace{-2mm}\forall x \in X_\safe\setminus X_\reach,  \label{tv-bounded-4} \\
            &\RA(x, i) \leq \mathbb{E}[\RA(\px', i-1)|x]  && \nonumber \\
            &  \quad - \alpha_{i-1} + \alpha_i + \beta_i 
                \!\!\! && \hspace{-2mm} \forall x \in X_\safe \! \setminus \! X_\reach,  \forall i \in \naturals^{\leq k}_{\geq 1} \label{tv-bounded-5} \\
            &\RA(x, k+1) \leq \mathbb{E}[\RA(\px', k)|x]  && \hspace{-2mm} \forall x \in X_\safe \! \setminus \! X_\reach \label{tv-bounded-6} \\
            &\RA(x, k+1) \geq \gamma && \hspace{-2mm} \forall x \in X_0. \label{tv-bounded-7} 
        \end{align}
    \end{subequations}
    Then, the  reach-avoid probability (for $H = \infty$) is lower-bounded by
    \begin{align}
        \label{eq: probability bound for TV infinite reach-avoid}
        \Pinf \geq \gamma - \left( \alpha_{k} + \sum_{i=1}^{ k} \beta_i \right ).
    \end{align}  
\end{corollary}
%

The proof follows in a similar fashion as in Theorem~\ref{th:reach-avoid value function}. In this formulation, the additional constraint \eqref{tv-bounded-6} ensures that the certificate reaches a steady-state form such that relaxation parameter $\beta_{k+1} = 0$, hence extending the guarantees to infinite horizon.

The time-varying certificates in Theorem~\ref{th:reach-avoid value function} and Corollary~\ref{corollary:reach-avoid infinite horizon} provide conservative (yet tighter) approximations of the reach-avoid probability (value function) at each time step. From an optimization standpoint, however, they may lead to a large problem, as a certificate with time-varying decision variables must be synthesized, i.e., number of decision variables increases with time steps. In many applications, it is preferable to work with a single certificate. To this end, we introduce a corresponding formulation for a time-invariant certificate for both finite- and infinite-horizon settings.


\begin{corollary}[Time-Invariant Reach-Avoid Certificate]
\label{col:reach-avoid time-invarying}
    Consider the setting in Theorem~\ref{th:reach-avoid value function}.
    Then, the continuous function $\RA : \mathbb{R}^n \to \mathbb{R}_{\geq 0}$ is a \emph{reach-avoid} certificate if, for constants $\alpha,\beta, \gamma \in \reals_{\geq 0}$, it satisfies 
    \begin{subequations}
        \begin{align}
            &\RA(x) \geq 0 && \forall x \in \mathbb{R}^n, \label{tiv-constraint1} \\
            &\RA(x) \leq 1 + \alpha && \forall x \in X_\reach,  \label{tiv-constraint2}\\
            &\RA(x) \leq \alpha && \forall x \in X_\unsafe,  \label{tiv-constraint3} \\
            &\RA(x) \leq \mathbb{E}[\RA(\px') \mid x] + \beta && \forall x \in X_\safe \setminus X_\reach, \label{tiv-constraint4-finite}\\
            &\RA(x) \geq \gamma && \forall x \in X_0.  \label{tiv-constraint5}
        \end{align}
    \end{subequations}
    Then, the  reach-avoid probability for a finite horizon $H \in \mathbb{N}_0$ is lower-bounded by
    \begin{align}
        \label{eq: probability bound for TI reach-avoid}
        \Pfin \geq \gamma - \alpha - \hor \beta.
    \end{align}  
    Furthermore, if $\RA$ satisfies Condition~\eqref{tiv-constraint4-finite} with $\beta = 0$,
        then the reach-avoid probability for infinite horizon $H = \infty$ is lower-bounded by
        \begin{align}
            \label{eq: probability bound for TI infinite reach-avoid}
            \Pinf \geq \gamma - \alpha.
        \end{align}
\end{corollary}

\begin{proof}
Similar to Theorem~\ref{th:reach-avoid value function}, the proof follows from DP. 
Recall $V_k(x)$ from~\eqref{eq: value function}.
We show by induction that $V_k(x) \ge \RA(x) - \alpha - k\beta$. For $k=0$, the inequality holds by the bounds on $\RA$. Assume it holds for some $k \ge 0$. 
Then, for $k+1$, Conditions~\eqref{tiv-constraint1}–\eqref{tiv-constraint3} hold trivially in a formal sense; therefore, we focus on the expectation, i.e., for $x \in X_\safe \setminus X_\reach$,
\begin{align*}
   V_{k+1}(x) & = \mathbb{E}[V_k(x') \mid x], \\
   & \ge \mathbb{E}[\RA(x') - \alpha - k\beta \mid x], \\
   & = \mathbb{E}[\RA(x') \mid x] -  \alpha - k\beta, \\
   & \ge \RA(x) - \alpha - (k+1) \beta,
\end{align*}
where the last step uses the inequality in \eqref{tiv-constraint4-finite}. 
Hence, for $H < \infty$, it holds that:
\begin{align*}
\Pfin & = \inf_{x \in X_0} V_H(x) \\ 
& \ge \inf_{x \in X_0} (\RA(x) - \alpha - \hor \beta) \\
& \ge \gamma - \alpha - \hor \beta.
\end{align*}
When $\beta=0$, the results trivially extends to $H = \infty$.
\end{proof}



\begin{remark}
    The super-martingale certificate in \cite{vzikelic2023learning} is a special case of our time-invariant certificate in Corollary~\ref{col:reach-avoid time-invarying}. Specifically, for the infinite-horizon case $H = \infty$ (with $\beta = 0$) and without relaxation (i.e., $\alpha = 0$), our certificate $\RA(x)$ recovers the certificate $V(x)$ in~\cite{vzikelic2023learning} through the relationship
    $\RA(x) = 1 - (1-\gamma)V(x)$
    with the same probably bound in~\eqref{eq: probability bound for TI infinite reach-avoid}.
\end{remark}




In what follows, we explain why relaxing the conditions by $\alpha$ allows for easier functional optimization, in particular for solvers such as SOS.

\begin{remark}
    \label{rem:alpha}
    Condition~ \ref{tiv-constraint1} requires the certificate to be non-negative globally. Further, condition \ref{tiv-constraint3}, without relaxation  requires $R(x) \leq 0$. The two conditions combined imply that $R(x) = 0$ must hold $\forall x \in X_u$. This condition is extremely hard to satisfy for continuous functions. Any infinitesimal (numerical) perturbation of the function will violate the constraints. Consequently, enforcing exact equality often renders the problem infeasible for parameterizations based on continuous polynomials. This problem is overcome by introducing variable $\alpha$. 
\end{remark}


With these sufficient conditions for reach-avoid certificates, our next goal is to construct an optimal $\RA$ that maximizes a lower bound on the probability. 
This is a functional optimization problem, where for time-varying, bounded-horizon certificates, the objective is
\begin{align}
    \label{eq: functional opt time-varying}
    \sup_{\RA} \; \gamma - \Big( \alpha_H + \sum_{i=1}^{H} \beta_i \Big) \quad \text{subject to \quad \eqref{eq:nonneg-tv}-\eqref{eq:Vinitial_ss}},
\end{align}
and for unbounded-horizon ($H = \infty$), time-varying certificates, it suffices to enforce $\beta_i + \alpha_i = 0$ beyond some finite step $k$, through Constraint~\eqref{tv-bounded-6}.  
For a finite-horizon, time-invariant certificate, this problem can be formulated as
\begin{align}
    \label{eq: functional opt time-invariant}
    \sup_{\RA} \;  \gamma - \alpha - H\beta  \quad \text{subject to \quad \eqref{tiv-constraint1}–\eqref{tiv-constraint5}}.
\end{align}
For the infinite-horizon case ($H = \infty$), it suffices to set $\beta = 0$ in both the objective function and Constraint~\eqref{tiv-constraint4-finite}.
Similarly, for controller synthesis, the supremum is taken over both $\RA$ and $\pi$.

These optimization problems are generally intractable due to their functional nature. In the next section, we show how they can be transformed into convex programs by restricting $\RA$ and $\pi$ to polynomial functions and employing SOS relaxations.

\section{Sum-of-Squares Functional Optimization}



Our approach to solving the reach–avoid verification and controller synthesis problems is based on SOS programming. Functional optimization via SOS enables the direct search for certificate functions within a polynomial function space. Accordingly, in this section we restrict $\RA$ to the class of polynomial functions.

\subsection{Certificate Synthesis via SOS Polynomials}


We first focus on certificate synthesis for a given controller~$\pi$ (Problem~\ref{Prob:ReachAvoidCertification}).  
Then, we formulate the time-varying certificate (Theorem~\ref{th:reach-avoid value function}) optimization problem in~\eqref{eq: functional opt time-varying} as an SOS program.
All time and state constraints are incorporated simultaneously, and the search for a valid certificate is a single convex optimization problem.

\begin{theorem}[Time-Varying, Bounded-Horizon Reach-Avoid SOS Certificate]
    \label{th:one-shot-certificate}
    Consider a sequence of SOS polynomial functions $\{\RA(x, i)\}_{i=0}^H$, and the semi-algebraic sets {$X_{\bar \safe} = X_{\safe} \setminus X_{\reach} = \{ x \in \reals^n \mid {h_{\bar s}(x)}  \ge 0\}$},
    $X_\initial = \{ x \in \reals^n \mid h_\initial(x) \ge 0\}$, $X_\unsafe = \reals^n \setminus X_\safe = \{ x \in \reals^n \mid h_\unsafe(x) \ge 0\}$, and $X_\reach = \{ x \in \reals^n \mid h_\reach(x) \ge 0\}$, where $h_j$ with $j\in \{\bar{s}, 0, \unsafe, \reach\}$ is a vector of polynomials. Let $\mathcal{L}_{\bar \safe}^{i}(x)$, $\mathcal{L}_\initial(x)$, $\mathcal{L}^{i}_\unsafe(x)$, and $\mathcal{L}^{i}_\reach(x)$ be vectors of SOS polynomials with the same dimensions as {$h_{\bar s}$}, $h_\initial$, $h_\unsafe$, and $h_\reach$, respectively, for each $i \in \naturals^{\leq \hor}_{\geq 0}$. 
    Also, let $\nx = f(x, \pi, \textbf{w})$ for a given 
    polynomial
    controller $\pi$. 
    Then, a \emph{time-varying reach-avoid certificate} can be obtained by solving the SOS optimization problem:
    \begin{subequations}
        \label{pr:sos-1}
        \begin{align}
        & \!\!\! \max_{\substack{\gamma, (\alpha_i)_{i=0}^H, (\beta_i)_{i=1}^H}} \!\!        \gamma - \left( \alpha_H + \sum_{i=1}^{ \hor} \beta_i
        \right )  
         \text{subject to:} \nonumber \\
        &\RA(x, i)\in \Lambda && \hspace{-10mm} \forall i \in \naturals^{\leq \hor}_{\geq 0} \label{constraint1} \\
        - & \RA(x, i) + 1 + \alpha_i -  \mathcal{L}^{i}_{\reach}(x)^{\top}h_{\reach}(x) \in \Lambda && \hspace{-10mm} \forall i \in \naturals^{\leq \hor}_{\geq 0}  \label{constraint2} \\
        - & \RA(x, i) + \alpha_i - \mathcal{L}^{i}_{\unsafe}(x)^{\top}h_{\unsafe}(x) \in \Lambda && \hspace{-10mm} \forall i \in \naturals^{\leq \hor}_{\geq 0}  \label{constraint3} \\
        - & \RA(x, 0) + \alpha_0 - \mathcal{L}^{0}_{\bar \safe}(x)^{\top} 
        {h_{\bar s}(x)} \in \Lambda  \label{constraint4}  \\
        - & \RA(x, i) + \mathbb{E}[\RA(\px', i-1) \mid x ] - \! \alpha_{i-1} + \nonumber \\
        & \hspace{13mm} \alpha_i +\beta_i -  \mathcal{L}^{i}_{\bar \safe}(x)^{\top}{h_{\bar s}(x)} \in \Lambda && \hspace{-10mm}\forall i \in \naturals^{\leq \hor}_{\geq 1}  \label{constraint5} \\
       &  \RA(x, \hor) - \gamma + \mathcal{L}_{\initial}(x)^{\top}h_{\initial}(x) \in \Lambda && \label{constraint6}  
        \end{align}
    \end{subequations}
    which guarantees the probabilistic reach-avoid bound in~\eqref{eq: probability bound for TV reach-avoid}.
\end{theorem}
\begin{proof}
It suffices to show that if $\RA$ satisfies constraints~\eqref{constraint1}-\eqref{constraint6}, then conditions~\eqref{eq:nonneg-tv}-\eqref{eq:Vinitial_ss} are satisfied. Since constraint~\eqref{constraint1} guarantees that each $\RA(x, i)$ is a SOS polynomial, the nonnegative condition~\eqref{eq:nonneg-tv} holds. Similarly, by Corollary~\ref{cor:sosp1}, if constraints \eqref{constraint2}-\eqref{constraint4} hold, then $\RA(x, i)$ is respectively smaller that $(1 + \alpha_i)$ in $X_r$, and $\RA(x, i)$ is smaller than $\alpha_i$ in $X_u$, for each $i = \left \{0, \dots, H \right \}$. For $i=0$, if~\eqref{constraint4} holds, then, by the same corollary condition~\eqref{eq:initial-nsafe-tv} is satisfied. Further, constraint\eqref{constraint5} poses the expectation condition as defined by Condition~\eqref{eq:exp-tv}. Since $\mathbf{x'}$ is composed of polynomials, the composition $\mathbb{E}[\RA(\mathbf{x'}, i-1) \mid x]$ yields a SOS constraint. Finally, constraint~\eqref{constraint6} satisfies the final horizon condition~\eqref{eq:Vinitial_ss}.
\end{proof}

\begin{remark}
    The formulation in Theorem~\ref{th:one-shot-certificate} enforces $\RA(x,i)$ to be an SOS polynomials.  This can be easily relaxed to the set of non-negative polynomials by replacing Constraint~\eqref{constraint1} with 
    $\RA(x, i) -  \mathcal{L}^{i}_{X}(x)^{\top}h_{X}(x) \in \Lambda$, where $\mathcal{L}^{i}_{X}(x)$ are vectors of SOS polynomials (Lagrange multipliers) and $h_{X}(x)$ is a vector of polynomials defining the semi-algebraic set $X = X_s \cup X_r \cup X_u$.
\end{remark}

For the synthesis of the time-invariant certificate in Corollary~\ref{col:reach-avoid time-invarying},
we obtain a much simpler SOS formulation below.


\begin{corollary}[Time-Invariant, Bounded Horizon Reach-Avoid SOS Certificate]
    \label{col:one-shot-certificate-ti}
    Consider the setting in Theorem~\ref{th:one-shot-certificate} with a time-invariant SOS polynomial function $\RA$.
    Then, an finite-horizon \emph{reach-avoid certificate} can be obtained by solving the following SOS optimization problem 
    \begin{subequations}
        \label{pr:sos-2}
        \begin{align}
        & \hspace{5mm} \max_{\substack{\gamma, \alpha, \beta}}
        \quad \gamma - (\alpha + H\beta)
        \quad 
         \text{subject to:} \nonumber \\
        &\RA(x)\in \Lambda &&  \\
        - & \RA(x) + 1 + \alpha -  \mathcal{L}_{\reach}(x)^{\top}h_{\reach}(x) \in \Lambda  \\
        - & \RA(x) + \alpha - \mathcal{L}_{\unsafe}(x)^{\top}h_{\unsafe}(x) \in \Lambda   \\
        - & \RA(x) + \mathbb{E}[\RA(\px') \mid x ] +\beta  -  \mathcal{L}_{\safe}(x)^{\top}{h_{\bar s}}(x) \in \Lambda  \\
       &  \RA(x) - \gamma + \mathcal{L}_{\initial}(x)^{\top}h_{\initial}(x) \in \Lambda && 
        \end{align}
    \end{subequations}
    which guarantees the probabilistic reach-avoid bound in~\eqref{eq: probability bound for TI reach-avoid}.
\end{corollary}
The proof of this corollary is a trivial adaptation of Theorem~\ref{th:one-shot-certificate}.
Also, we can easily extend the SOS-optimization formulations in Theorem~\ref{th:one-shot-certificate} and Corollary~\ref{col:one-shot-certificate-ti} to the infinite horizon setting by setting $\beta_H$ and $\beta$ to $0$, respectively.  


As discussed in Section~\ref{sec: certificate formulations}, the time-varying and time-invariant certificates exhibit different computational properties. Roughly speaking, the time-varying formulation yields tighter lower bounds but incurs higher computational cost. To make this trade-off precise, we analyze the computational complexity of each problem below.

\subsubsection{Computational Complexity}
Both the time-varying (Theorem~\ref{th:one-shot-certificate}) and time-invariant (Corollary~\ref{col:one-shot-certificate-ti}) formulations use SOS polynomials of degree $d$ in $n$ variables, with $b(d,n)=\binom{n+d}{d}$ monomials. Each SOS constraint yields an SDP block of size $\mathcal{O}(b(d,n))$, with $\mathcal{O}(b(d,n)^2)$ variables and computational cost $\mathcal{O}(b(d,n)^3)$.  We can analyze the runtime complexity of each certificate formulation, as detailed below.

\paragraph{Time-varying certificate in Theorem~\ref{th:one-shot-certificate}}
The finite-horizon formulation introduces $(H{+}1)$ polynomials $\{\RA(x,i)\}_{i=0}^H$ and corresponding multipliers.
This yields $\approx 10 \times (H{+}1)$ SOS constraints in total,
resulting in an SDP with (roughly) $(H{+}1)$ blocks of size $b(d,n)$. The overall worst-case complexity is 
$$\mathcal{O}\!\big((H{+}1)^3 b(d,n)^3\big).$$

\paragraph{Time-invariant certificate in Corollary~\ref{col:one-shot-certificate-ti}}
The single-certificate formulation uses one polynomial $\RA(x)$ and one set of multipliers, yielding an SDP of size $b(d,n)$ independent of the horizon, with a worst-case complexity 
$$\mathcal{O}\!\big(b(d,n)^3\big).$$

\paragraph{Comparison}
%
The time-varying formulation is more computationally expensive by a factor of $(H{+}1)^3$, 
since it essentially requires synthesizing a different certificate at every time step. The time-invariant approach avoids this cost by reusing a single certificate across the horizon, at the expense of increased conservatism. 
Determining which method offers the best trade-off between tightness of the probability bound and computational cost is problem dependent.  

Furthermore, recall from Remark~\ref{rem:alpha} that the constraints effectively require equality conditions to be satisfied, which are difficult to enforce using continuous parameterizations. As such, relaxations $\alpha_i$ are introduced to ensure feasibility of the optimization problem, which introduces an approximation error due to representing the reach–avoid value function using continuous certificates. 
To reduce the effect of this conservatism, one option is to use higher-degree polynomials for $\RA$ in the time-invariant setting, which can decrease conservatism while remaining tractable. 
The time-varying formulation can achieve even tighter approximations, but with higher computational cost. 
In practice, as observed in our benchmarks in Section~\ref{sec:experiments},
using lower-degree polynomials in the time-varying formulation often provides a favorable balance, offering improved approximation accuracy while keeping the optimization manageable.  

In other words, a single time-invariant certificate often requires a highly expressive (high-degree) polynomial template, whereas time-varying certificates can achieve similar or even tighter probability bounds using multiple low-degree polynomials.  Finally, we note that a promising direction for future work is to investigate decomposition techniques that reduce the time-varying SOS optimization problem into a sequence of smaller, more tractable problems.  This is possible because the time aspect of the certificate can essentially be viewed as a separate optimization problem.



\subsection{Joint synthesis of Controller and Certificate}
In this section, we focus on Problem~\ref{Prob:ReachAvoidSynthesis} and present a novel SOS-optimization problem for the joint synthesis of the certificate $\RA$ and the feedback controller $\pi$.  As such, in what follows, we assume $\pi$ is a polynomial function of $x$.

We first establish two central lemmas that our approach is built on.
The first lemma, which is an adaptation of \cite{bach2025sum}, establishes how a \emph{min-max} over polynomials can be relaxed to an SOS program.  For simplicity of presentation, we focus on time-invariant $\RA$ and $\pi$ in an infinite-horizon setting.

\begin{lemma}[SOS Relaxation for Min-Max Expectation]
\label{lemma:dual}
Let real-valued function $\RA : \mathbb{R}^n \to \mathbb{R}_{\geq 0}$ and $\nx = f(x, \pi(x), \textbf{w})$, such that the expectation condition in~\eqref{tiv-constraint4-finite}, for $\beta = 0$ (infinite-horizon) is
\[
g(\pi(x),x)\; := \;\mathbb{E}[\RA(\px')\mid x,\pi] -\RA(x),
\]
where $g$ is a polynomial in the controller parameters $\pi(x)$ and the state $x$.  
Further, recall semi-algebraic set 
${X_{\bar \safe}} = X_\safe\setminus X_\reach = \{\, x\in\mathbb{R}^n \mid {h_{\bar s}}(x)\ge 0 \,\},$
where ${h_{\bar s}}$ is a vector of polynomials. 
Then, the synthesis constraint
\begin{align}
    \label{eq: min-max problem}
    \min_{\pi}\;\max_{x\in {X_{\bar \safe}}}\; g(\pi(x),x) \;\ge\; 0
\end{align}
is a polynomial min--max problem on set ${X_{\bar \safe}}$. 
The admissible controller parameters are restricted to a compact semi-algebraic set $ U := \{ u \in \mathbb{R}^m \mid h_U(u) \ge 0 \},$ where $h_U(u)$ is a vector of polynomials. The admissible controller then satisfies
$\pi(x) \in U.$
The synthesis constraint admits a level-$d$ moment/SOS relaxation~\cite{bach2025sum}. Express $g(\pi,x)$   as
\[
g(\pi(x),x)=\sum_{\eta} g_{\eta}(\pi(x))\,x^{\eta}.
\]
Then, the primal truncated moment relaxation is
\begin{subequations}
    \label{eq:primal-moment}
    \begin{align}
        J_{p} = \sup_{y,\sigma}\ & \sigma \\
        \text{s.t.}\quad
        & y_0 = 1, \\
        & M_d(y) \succeq 0,  \\
        & M_{d-d_h}\big({h_{\bar s}}  (y)\big) \succeq 0, \\
        & \sum_\eta g_\eta(\pi(x)) y_\eta \ge \sigma, \quad \forall \pi,
    \end{align}
\end{subequations}
where $M_d(y)$ is the moment matrix, and  $M_{d-d_h}\big({h_{\bar s}}  (y))$ denotes the localizing matrices associated with the polynomial constraints defining {$X_{\bar \safe}$} in the Lasserre hierarchy.
Then, dualizing the infinite constraint over $\pi$ with a Putinar-type SOS representation yields the following dual SOS program \cite{bach2025sum}:
\begin{subequations}
    \label{eq:dual-sos}
    \begin{align}
        J_{d} = \inf_{\lambda, s_0}\ & \lambda \\
        \text{s.t.}\quad &
        \lambda - g(\pi(x)) = s_0(\pi(x)), \\
        & s_0 \in \Lambda,
    \end{align}
\end{subequations}
where $\Lambda$ denotes the cone of SOS polynomials in $x$ up to a chosen degree $d$. If the relaxations \eqref{eq:primal-moment}--\eqref{eq:dual-sos} admits a feasible solution with value $\sigma^\star\ge 0$, then
\[
\min_{\pi\in U}\max_{x\in {X_{\bar \safe}} } g(\pi(x),x)\ge \sigma^\star\ge 0.
\]
\end{lemma}
\begin{proof}
For fixed \(\pi\in U\), the inner maximization can be expressed as
\[
\max_{x\in {X_{\bar \safe}} } g(\pi,x)
= \sup_{\mu \in \mathcal{P}({X_{\bar \safe}} )} \int_{{X_{\bar \safe}} } g(\pi,x)\, d\mu(x),
\]
where \(\mathcal{P}({X_{\bar \safe}} )\) is the set of probability measures supported on \({X_{\bar \safe}} \). Hence,
\[
\min_{\pi\in U}\max_{x\in {X_{\bar \safe}} } g(\pi,x)
= \min_{\pi\in U} \sup_{\mu \in \mathcal{P}({X_{\bar \safe}} )} \int g(\pi,x)\, d\mu(x).
\]
By weak duality,
\[
\min_{\pi\in U} \sup_{\mu\in\mathcal{P}({X_{\bar \safe}} )} \int g\, d\mu
\ge \sup_{\mu\in\mathcal{P}({X_{\bar \safe}} )} \min_{\pi\in U} \int g\, d\mu.
\]
Approximating \(\mu\) by its truncated moment sequence \(y\) subject to
\[
M_d(y)\succeq 0, \quad 
M_{d-d_h}\big({h_{\bar s}}  (y)\big) \succeq 0, \quad 
y_0 = 1,
\]
and enforcing
\[
\sum_\eta g_\eta(\pi)\, y_\eta \ge \sigma, \quad \forall \pi \in U,
\]
the problem of maximizing \(\sigma\) can be cast as the 
primal moment SDP \eqref{eq:primal-moment}, which provides a lower bound on the min–max value. Then, dualizing the infinite constraint over \(\pi\) via Putinar’s Positivstellensatz \cite{baldi2023effective} produces the SOS program \eqref{eq:dual-sos}.  
Hence, feasible solutions of \eqref{eq:primal-moment}--\eqref{eq:dual-sos} certify
\[
\min_{\pi \in U} \max_{x \in {X_{\bar \safe}} } g(\pi,x) \ge \sigma^\star \ge 0,
\]
establishing the synthesis constraint.
\end{proof}

Intuitively,~\eqref{eq:primal-moment} ensures that the moment variables represent a valid probability measure, while~\eqref{eq:dual-sos} certifies via SOS that $\lambda$ upper-bounds $g$ over the domain. This result shows that the nonconvex min–max problem in~\eqref{eq: min-max problem}, which jointly optimizes over controller parameters and worst-case states, can be relaxed into a convex  program via SOS relaxations. This  allows us to compute a controller and its corresponding reach-avoid certificate simultaneously. 

The following lemma shows that there is zero duality gap between the primal \eqref{eq:primal-moment} and dual \eqref{eq:dual-sos} under certain conditions.

\begin{lemma}[Zero Duality Gap]
\label{lemma:zero-dual}
    Let $X \subset\mathbb{R}^n$ be a compact set and assume the moment and localizing matrix constraints satisfy Slater's condition, i.e., there exists a strictly feasible point \cite{boyd2004convex}.
    Then, for any relaxation degree $d$, the primal moment SDP \eqref{eq:primal-moment} and the dual SOS program \eqref{eq:dual-sos} have zero duality gap, i.e., $J_p = J_d$.
Consequently, any feasible solution $\sigma^\star$ of the dual SOS program certifies a valid lower bound on the original min-max problem in~\eqref{eq: min-max problem}.
\end{lemma}
\begin{proof}
Since ${X_{\bar \safe}} $ is compact and the quadratic module generated by its defining polynomials is Archimedean (i.e., it contains a polynomial of the form $R - \|x\|^2$ for some $R>0$, which ensures boundedness of ${X_{\bar \safe}} $), Putinar's Positivstellensatz guarantees that any strictly positive polynomial on ${X_{\bar \safe}} $ admits an SOS representation \cite{nie2007complexity}.  
Combined with Slater's condition for the moment SDP, standard SDP duality \cite{boyd2004convex} implies that the primal and dual achieve the same optimal value at relaxation degree $d$.
\end{proof}

A direct consequence of this lemma is that, as the SOS relaxation degree $d\to\infty$, the SOS hierarchy (primal and dual) recovers the true min-max value of the problem in~\eqref{eq: min-max problem}.


We are ready to present our main synthesis theorem, that uses the proceeding lemmas to formulate one SOS program.

\begin{theorem}[Time-Invariant, Unbounded Horizon Reach-\\Avoid Controller Synthesis]
    \label{th:one-shot-synthesis}
    Consider the setting in Corollary~\ref{col:one-shot-certificate-ti}, where controller $\pi$ is given as a parameterized polynomial. 
    Then, optimal controller $\pi^*$ and its corresponding \emph{reach-avoid certificate} $\RA$ can be obtained by solving the following SOS optimization problem 
    \begin{subequations}
\label{pr:outer}
\begin{align}
& \hspace{5mm} \max_{\gamma, \alpha,\pi} \quad \gamma - \alpha \quad 
         \text{subject to:} \nonumber \\
& \RA(x) \in \Lambda, \\
- & \RA(x) + 1 + \alpha - \mathcal{L}_\reach(x)^\top h_\reach(x) \in \Lambda, \\
 - & \RA(x) + \alpha - \mathcal{L}_\unsafe(x)^\top h_\unsafe(x) \in \Lambda, \\
& \RA(x) - \gamma + \mathcal{L}_\initial(x)^\top h_\initial(x) \in \Lambda, \\
& g^\star(x) := \min_{\pi(x) \in U} g(\pi(x),x) \ge 0, \qquad \qquad \forall x \in {X_{\bar \safe}},   \label{eq: inner opt constraint}
\end{align}
\end{subequations}
where the inner problem in \eqref{eq: inner opt constraint} is the dual SOS synthesis from Lemma~\ref{lemma:dual}, 
\begin{subequations}
\label{pr:inner-sos}
\begin{align}
g^\star(x) = \inf_{\lambda, s_0} \ & \lambda \\
\text{s.t.} \quad 
& \lambda - g(\pi(x)) = s_0(\pi(x)), \\
& s_0 \in \Lambda.
\end{align}
\end{subequations}
A feasible solution to this optimization guarantees the probabilistic reach-avoid bound in~\eqref{eq: probability bound for TI reach-avoid}.
\end{theorem}

\begin{proof}
The result follows directly from Theorem~\ref{th:one-shot-certificate}, which establishes the SOS reach-avoid certificate conditions, together with application of Lemmas~\ref{lemma:dual} and~\ref{lemma:zero-dual}, which ensure that the inner min--max problem can be replaced by a dual SOS program, and zero duality gap, respectively. 
\end{proof}


\begin{remark}
    The synthesis procedures for the time-varying finite-and infinite-horizon cases follow analogously to Theorem~\ref{th:one-shot-synthesis}.
    Furthermore, for infinite-horizon problems, controller $\pi$ suffices to be time-invariant for optimality.  However, for finite-horizon problems with time-varying certificates, $\pi$ needs to be time-varying to guarantee optimality.
\end{remark}




The above controller synthesis problem involves solving an inner optimization program. This nested structure increases the computational complexity compared to the verification problem, since each certificate requires solving a dual SOS problem over the admissible control set.

\section{Experiments}
\label{sec:experiments}

In this section, we evaluate the proposed reach-avoid certificates and controller synthesis on a range of stochastic systems exhibiting linear or polynomial dynamics. 
The first set of experiments presents a benchmark comparison with the frameworks in \cite{vzikelic2023learning} and \cite{xue2025finite}, corresponding to infinite- and finite-horizon certificates, respectively, to illustrate the effectiveness of our formulation.
In particular, we show that $\alpha$-relaxations are more suitable for SOS optimization than the one in \cite{vzikelic2023learning}. Then, we conducted two categories of experiments: verification (certificate synthesis) and controller synthesis (joint synthesis of controller and certificate). 
In the synthesis experiments, a linear feedback template $\pi(x) = -K x$ was used. 

We implemented our algorithms in Julia and used the \texttt{SumOfSquares.jl} \cite{weisser2019polynomial} and \texttt{JuMP} \cite{dunning2017jump} packages for SOS optimization. All experiments were executed on a machine equipped with a 3.9 GHz 8-core CPU and 128 GB of RAM. A summary of the results is provided in 
Tables~\ref{tab:benchmark3}-\ref{tab:reach-avoid-2}.
In what follows, we provide a description of the considered systems, followed by a presentation of the corresponding results.

\subsection{Benchmark Systems}
\label{sec:benchsys}

\subsubsection{1D Linear Systems}
To showcase the intuition of each approach, we consider two one-dimensional systems with stable and unstable dynamics, both with $\pw_k \sim \N(0, 0.1)$. The dynamics of the first system are
\begin{align}
    \label{eq:1D_linear_stable}
    \px_{k+1} = 0.75 \px_k + \pw_k,
\end{align}
with \( X_{\safe} = [1, 3] \), \( X_{\initial} = [1.8, 1.9] \), \( X_{\reach} = [0.7, 1.0] \), and \( X_{\unsafe} = [3.0, 3.3] \).
The second system is
\begin{align}
    \label{eq:1D_linear_unstable}
    \px_{k+1} = 1.15 \px_k + \pw_k,
\end{align}
with \( X_{\safe}  = [1, 3] \), \( X_{\initial} = [2.2, 2.3] \),
\( X_{\reach} = [3.0, 3.5] \), and \(  X_{\unsafe}= [0.7, 1.0] \). 
For both the above systems, the control $U \in [-0.25, 0.25]$ for linear feedback control synthesis.

\subsubsection{2D Contraction Maps}

Next, we consider two \emph{contraction} systems in $\reals^{2}$ with a rotation. For the first contraction map, the dynamics are given as
\begin{align}
    \label{eq:contraction_map_1}
    \px_{k+1} =
    \begin{bmatrix}
        0.75 & 0.15\\
        -0.15 & 0.50
    \end{bmatrix}\px_k + \pw_k,
\end{align}
with $X_{\safe} = [-0.2, 2]\times[-0.5, 0.5],~X_{\initial} = [1.6, 1.7]\times[-0.1, 0.1],$, $X_{\reach} = [-0.2, 0.2]\times[-0.5, 0.5],$ and the unsafe region is around the safe set.  The noise is $\pw_k \sim \N(0, 10^{-2} I)$. The environment and 1000 Monte Carlo simulations are presented in Fig.~\ref{fig:mc1}. The empirical probability of reach-avoid for this scenario is $P_\reach = 0.96$, for $H=100$ steps.

The second case study has dynamics
\begin{align}
    \label{eq:contraction_map_2}
    \px_{k+1} =
    \begin{bmatrix}
        0.50 & 0.55\\
        -0.35 & 0.50
    \end{bmatrix}\px_k + \pw_k,
\end{align}
with the setup being the same as the previous case. For this system, if uncontrolled, the Monte Carlo probability of reach is $P_\reach = 0.19$. The simulation is provided in Fig.~\ref{fig:mc2}. The control space is bounded $U \in [-0.25, 0.25]^2$, for linear feedback control synthesis.

\begin{figure*}[t]
    \centering
    \begin{subfigure}[t]{0.32\textwidth}
        \includegraphics[width=\linewidth]{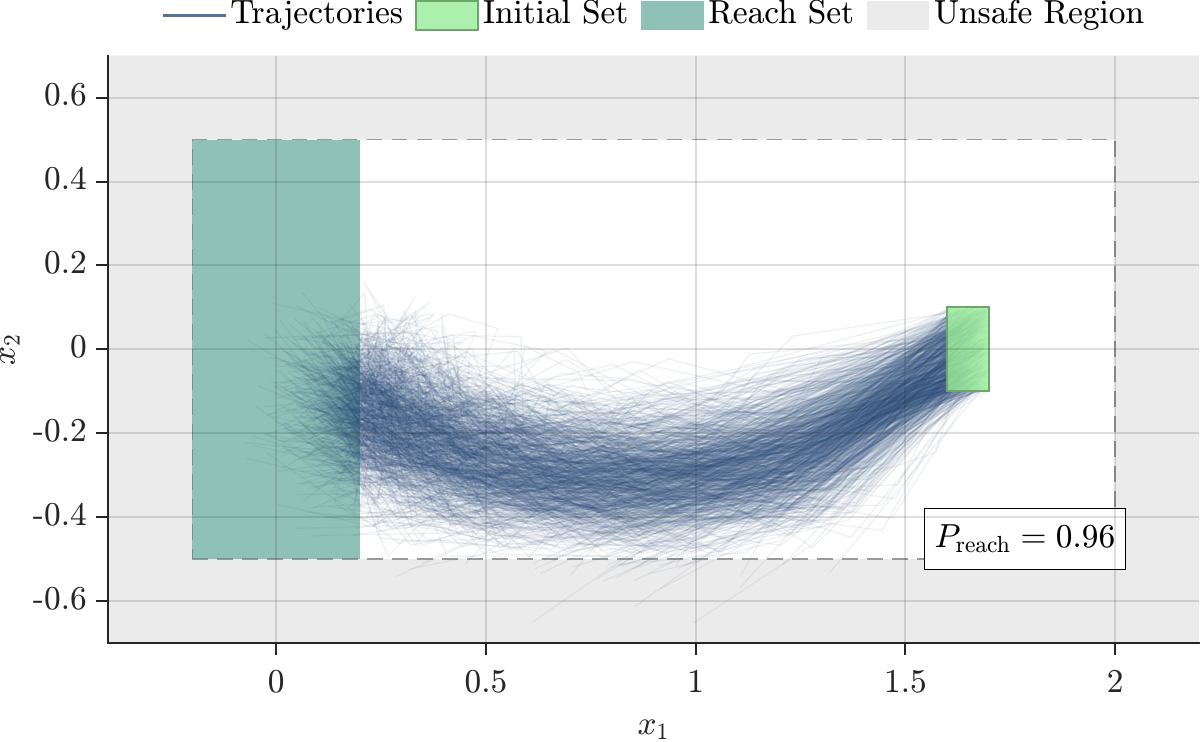}
        \caption{Contraction Map 1: empirical $P_r = 0.96$ }
        \label{fig:mc1}
    \end{subfigure}
    \hfill
    \begin{subfigure}[t]{0.32\textwidth}
        \includegraphics[width=\linewidth]{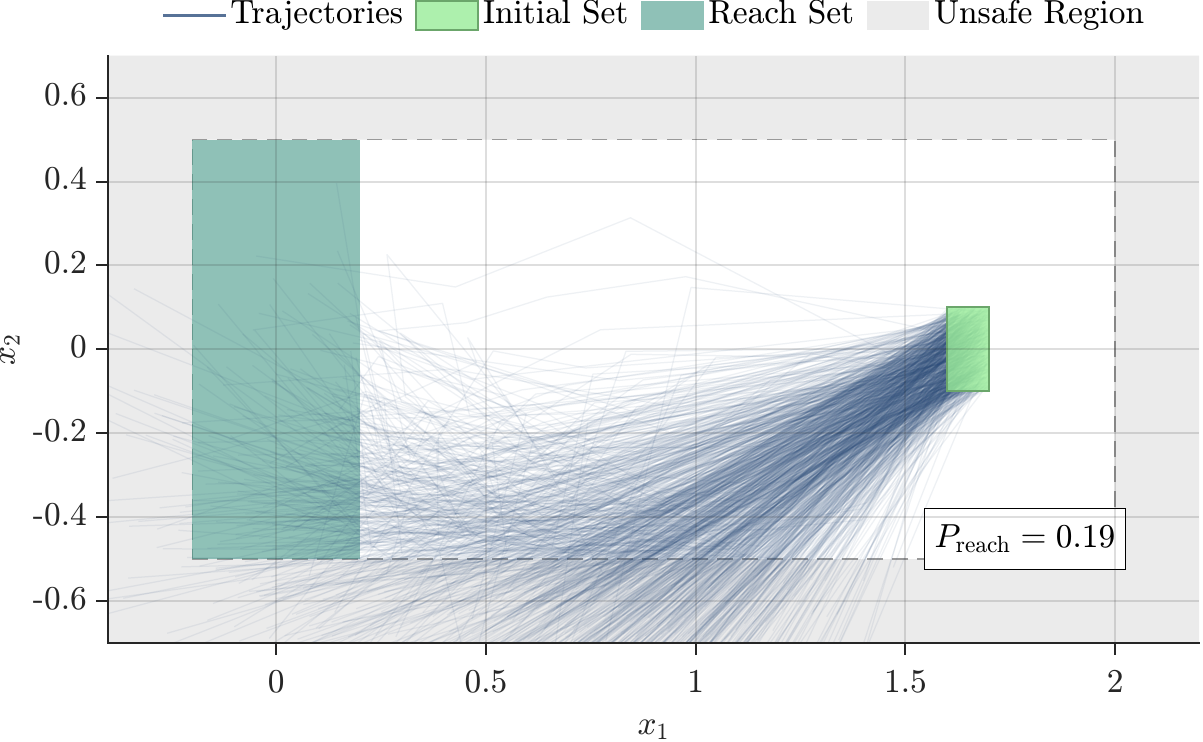}
        \caption{Contraction Map 2: empirical $P_r = 0.19$}
        \label{fig:mc2}
    \end{subfigure}
    \hfill
    \begin{subfigure}[t]{0.32\textwidth}
        \includegraphics[width=\linewidth]{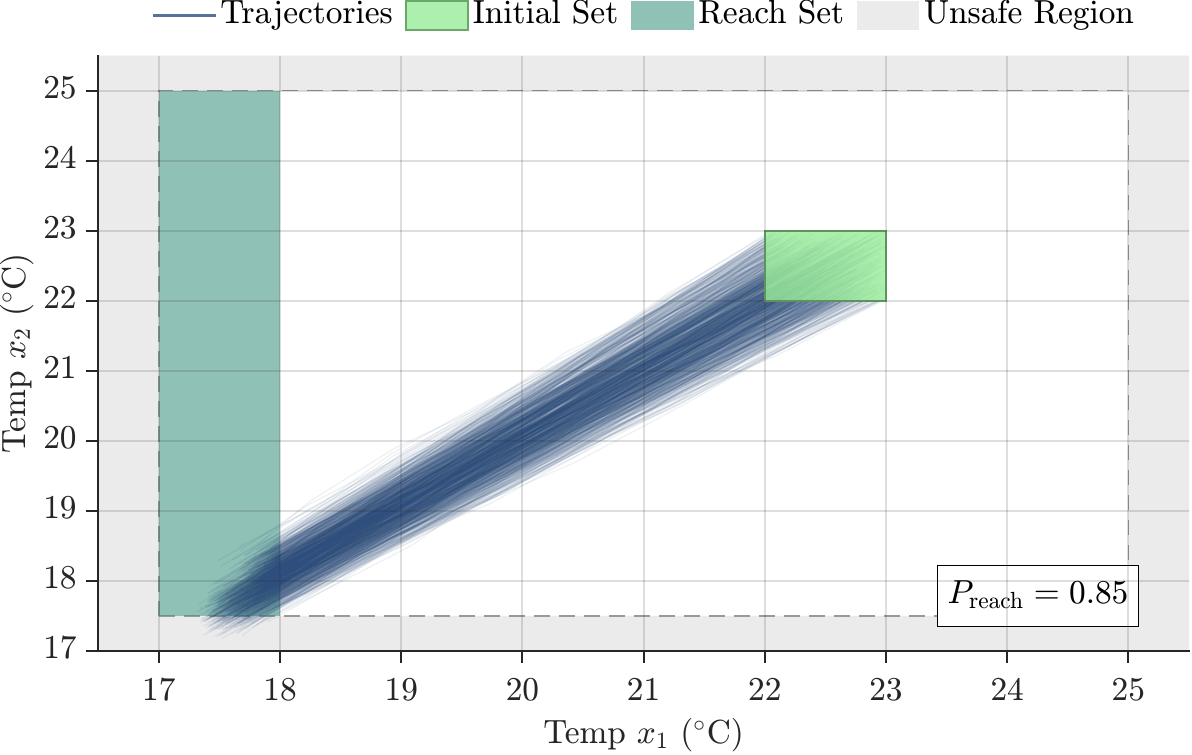}
        \caption{Room temperature: empirical $P_r = 0.85$}
        \label{fig:mc3}
    \end{subfigure}
    \caption{Monte Carlo simulations for different linear and polynomial verification systems.}
    \label{fig:mc-verification}
\end{figure*}
\begin{figure*}[t]
    \centering
    \begin{subfigure}[t]{0.32\textwidth}
        \includegraphics[width=\linewidth]{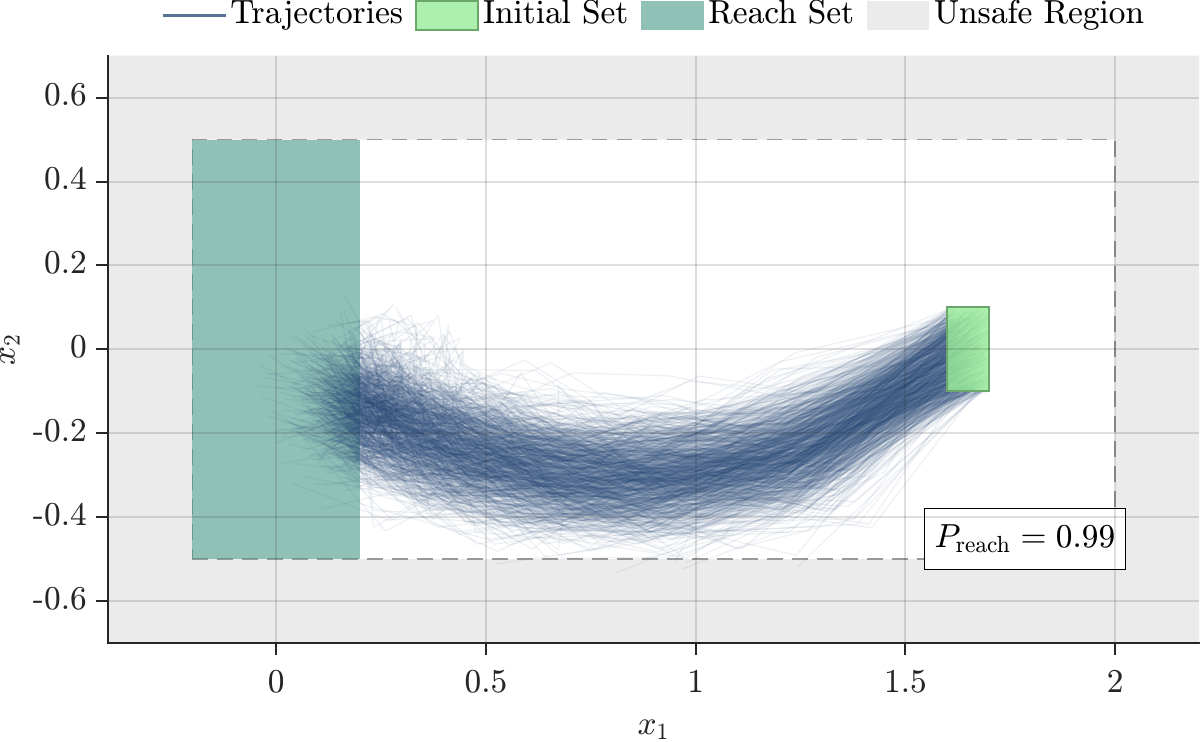}
        \caption{Contraction Map 1: empirical $P_r = 0.994$}
        \label{fig:mc1-c}
    \end{subfigure}
    \hfill
    \begin{subfigure}[t]{0.32\textwidth}
        \includegraphics[width=\linewidth]{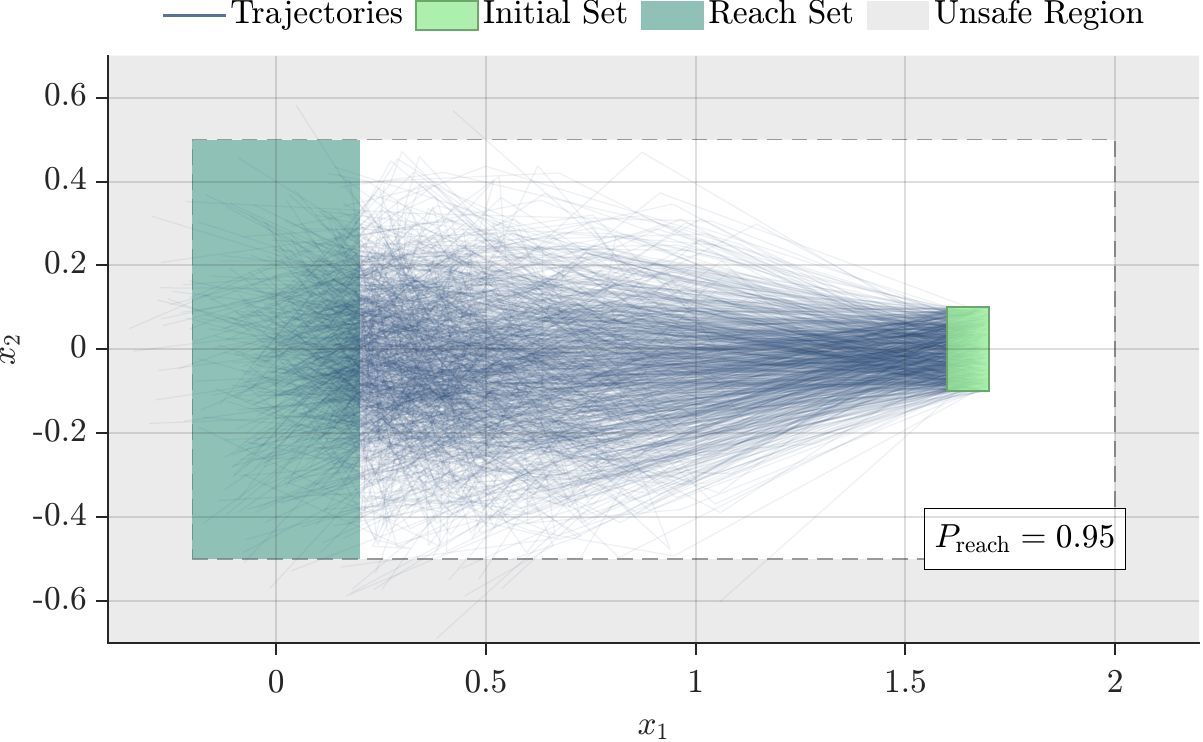}
        \caption{Contraction Map 2: empirical $P_r = 0.95$}
        \label{fig:mc2-c}
    \end{subfigure}
    \hfill
    \begin{subfigure}[t]{0.32\textwidth}
        \includegraphics[width=\linewidth]{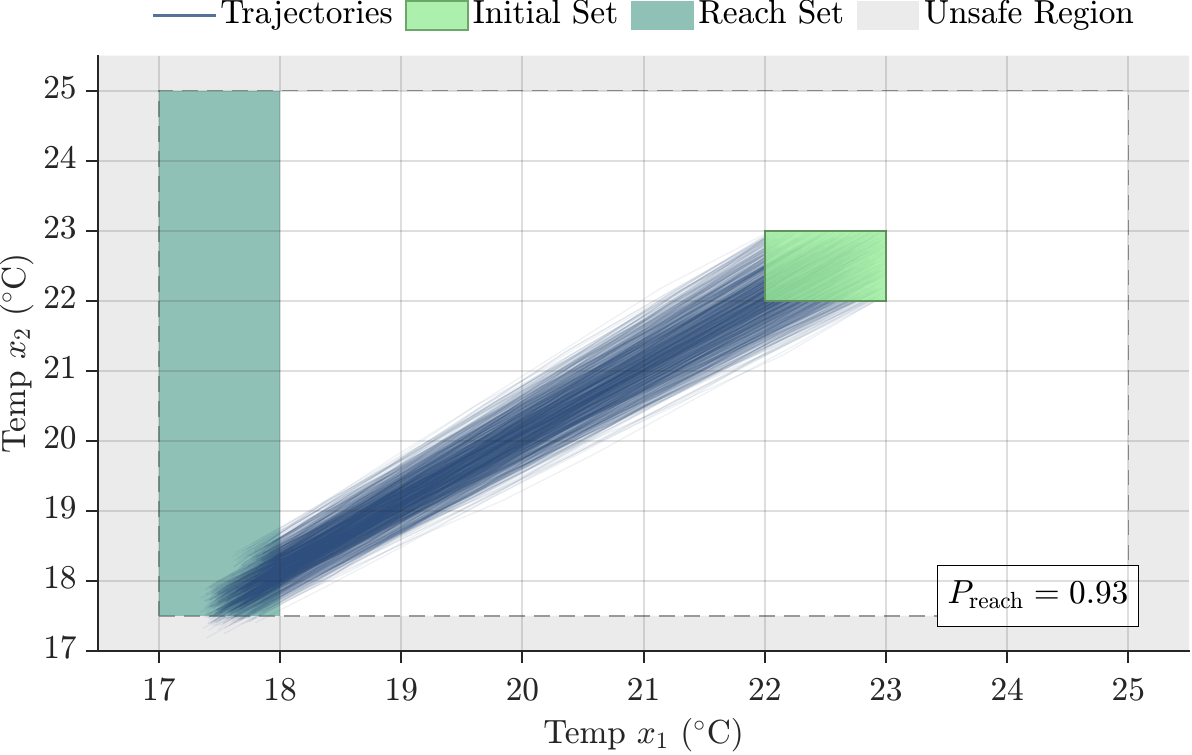}
        \caption{Room temperature: empirical $P_r = 0.93$}
        \label{fig:mc3-c}
    \end{subfigure}
    \caption{Monte Carlo simulations under synthesized controllers.}
    \label{fig:mc-synthesis}
\end{figure*}

\subsubsection{3D Aircraft Model}
\label{sec:aicraft}
We also consider a linearized aircraft model in $\reals^{3}$, with continuous-time longitudinal dynamics
{\footnotesize
\begin{align*}
    \dot{\px} =
    \begin{bmatrix}
        X_u/m & X_w/m & -g\cos\theta_0/m\\
        Z_u/(m-Z_{\dot w}) & Z_w/(m-Z_{\dot w}) & (Z_q+mU_0)/(m-Z_{\dot w})\\
        M_u/I_{yy} & M_w/I_{yy} & M_q/I_{yy}
    \end{bmatrix} \px,
\end{align*}
}
where $\mathbf{x} = [u, w, q]^\top$ denotes the forward velocity, vertical velocity, and pitch rate perturbations. Further, $m$ denotes the aircraft mass, and $I_{yy}$ the pitch-axis moment of inertia. We used a forward Euler scheme with $\Delta t = 0.1$ to discretize time, and use a LQR controller to stabilize the system near the reach set. Then noise $\pw_k \sim \N(0, 10^{-2}I)$ was added to the system. 
All parameter values of this model are provided in Appendix~\ref{parameters}.
The empirical reach-avoid probability was $P_\reach = 0.99$.

\subsubsection{2D \& 3D Room Temperature}
We finally consider a 2D and a 3D room temperature system with polynomial dynamics.  The 2D model is
{\small
\begin{align*}
    \textbf{x}_1(k+1) & = (1 - \tau(\alpha + \alpha_e)) \px_1(k) + \tau \alpha \px_2(k) + \tau \alpha_e T_e + \pw_1(k) \\
    \textbf{x}_2(k+1) & = (1 - \tau(2\alpha + \alpha_e)) \px_2(k) + \tau \alpha \px_1(k) + \tau \alpha_e T_e + \pw_2(k)
\end{align*}
}
where $\pw \sim \N(0, 10^{-2}I)$. The states $\px_1$ and $\px_2$ define the temperature of each room. The ambient temperature $T_e = 10^\circ$C, while the heat coefficients are $\alpha_e = 8 \times 10^{-3}$ and $\alpha = 6.2 \times 10^{-3}$. The sampling time $\tau$ = 5 minutes. The regions of interest are depicted in Fig.~\ref{fig:mc3}, and the empirical probability of reach-avoid is $P_\reach = 0.85$. The controller for the linear feedback controller is allowed to range between $U = [-0.5, 0.5]$. 
We also consider an equivalent of this system in $\reals^{3}$, which has an empirical $P_\reach = 0.82$.

\subsection{Results and Discussions}
{We are finally ready to present and discuss our results. We first present a comparison of our time-invariant certificate against the one in \cite{vzikelic2023learning} (Table~\ref{tab:benchmark3}), followed by comparison of our time-invariant and time-varying formulations in verification (Table~\ref{tab:reach-avoid-1}) and control synthesis (Table~\ref{tab:reach-avoid-2}) problems over finite- and infinite-time horizons.}

\begin{table}[t]
\caption{Time-invariant certificate comparison against \cite{vzikelic2023learning}. \texttt{deg} is the degree of SOS polynomial, $\low{P}_\reach$ is the obtained lower-bound on $P_\reach$, and $\tau$ is the computation time in seconds.}
\label{tab:benchmark3}
\centering
\resizebox{\columnwidth}{!}{
\begin{tabular}{c | c | cr | cr }
\toprule
\toprule
\multirow{3}{*}{Model}  &
\multirow{3}{*}{\texttt{deg}} 
& \multicolumn{4}{c}{$\text{Time-Invariant Reach-Avoid Certificate}$}  \\
& & \multicolumn{2}{c}{\underline{\quad $V(x)$ (\!\!\cite{vzikelic2023learning}) \quad } } 
& \multicolumn{2}{c}{\underline{\quad $\RA(x)$ (ours) \quad}  } \\
& & $\low{P}_\reach$ & $\tau(s)$ \ 
& $\low{P}_\reach$ & $\tau(s)$ \ \\ 
\midrule

{1D Linear} \texttt{Map 1} & 8  & 0.16 & 0.28 & \textbf{0.76} & 0.18 \\
\eqref{eq:1D_linear_stable} & 10 & 0.17 & 1.13 & \textbf{0.96} & 1.85 \\

\midrule
{2D Contraction} \texttt{Map 1} & 10 & \texttt{FAILED} & 0.97 & \textbf{0.56} & 1.06 \\
$X_{\initial}^1$ & 14 & 0.15 & 2.19 & \textbf{0.91} & 2.11 \\
 & 20 & 0.43 & 18.69 & \textbf{0.96} & 21.09 \\

\midrule
{2D Contraction} \texttt{Map 1} & 14 & 0.42 & 2.25 & \textbf{0.86} & 2.14 \\
$X_{\initial}^2$ & 20 & 0.49 & 30.13 & \textbf{0.94} & 19.46 \\
 & 24 & 0.54 & 71.59 & \textbf{0.97} & 79.48 \\

\midrule
{2D Contraction} \texttt{Map 1} & 14 & 0.24 & 2.48 & \textbf{0.65} & 2.26 \\
$X_{\initial}^3$ & 20 & 0.34 & 23.86 & \textbf{0.81} & 19.23 \\
 & 24 & \texttt{FAILED} & 72.14 & \textbf{0.89} & 71.18 \\

\bottomrule
\end{tabular}
}
\end{table}

\begin{table*}[t]
\caption{Reach-avoid \texttt{verification} results for finite- and infinite-time horizons. \texttt{deg} is the degree of SOS polynomial, $\low{P}_\reach$ is the obtained lower-bound on $P_\reach$, $\tau$ is the computation time in seconds.
}
\label{tab:reach-avoid-1}
\centering
\scalebox{1}
    {
\begin{tabular}{c | cccr | ccr | cccr | cccr}
\toprule 
 \toprule  
  & & \multicolumn{6}{c}{$\text{Time-Invariant Certificate}$}  
   & \multicolumn{8}{c}{\text{Time-Varying Certificate}}  \\
 \multirow{ 2}{*}{Model} & 
 & \multicolumn{3}{c}{$\texttt{Finite-Horizon}$}  
   & \multicolumn{3}{c}{\texttt{Infinite-Horizon}}  
    & \multicolumn{4}{c}{$\texttt{Finite-Horizon}$}  
   & \multicolumn{4}{c}{\texttt{Infinite-Horizon}}  \\
   \\
  & $H$ & \texttt{deg} &   $\low{P}_\reach$ & $\tau(s) \ $ 
   & \texttt{deg} & $\low{P}_\reach$   & $\tau(s) \ $ 
  & $H$ & \texttt{deg} &   $\low{P}_\reach$ & $\tau(s) \ \ $ 
   &\texttt{k} &  \texttt{deg} & $\low{P}_\reach$   & $\tau(s) \ \ $ \\
 \midrule
  {1D Linear}  &  10   & 10  & 0.76 &  0.16  & 10 & 0.76 & 0.19 & 10 & 6 & 0.98 & 1.27 & 10 & 6 &1.0 & 2.52  \\
  \eqref{eq:1D_linear_stable}    &  10  & 20 & 0.99 & 0.19 & 20 & 0.99 & 0.27 & 10 & 8 & 0.99 & 2.91 & 10 & 8 & 1.0 & 2.68  \\ 
  &  20  & 10 & 0.76 & 0.17 & 10 &  0.76 & 0.23 & 20 & 6 & 0.99 & 2.70 & 20 & 6 & 1.0 &  4.31 \\ 
  &  20  & 20 & 0.99 & 0.20 & 20 & 0.99 & 0.22 & 20 &  8& 0.99 & 6.83 & 20 & 8 & 1.0 & 5.13\\ 
 \cline{2-16} 
    {1D Linear}  &    10 & 10 & 0.77 & 0.19 & 10 & 0.77 & 0.39 & 10 & 6 & 0.98 & 2.86 & 10 & 6 & 1.0 & 2.18 \\
 \eqref{eq:1D_linear_unstable}  & 10 &  20 & 0.93 & 0.79  & 20 & 0.92 & 0.33 & 10 & 8 & 0.99 & 4.06 & 10 & 8 & 1.0 & 2.85\\ 
  &   20  & 10 & 0.77 & 0.27 & 10 &0.77 & 0.50 & 20 & 6 & 0.98 &  2.59 & 20 & 6 & 1.0 & 4.54\\
    &  20   & 20  & 0.91 & 0.44 & 20 & 1.0 & 0.27 & 20 & 8 & 0.99 & 6.58 &20 &  8 & 1.0 & 5.56 \\
\midrule
     {2D Contraction}      &  10 & 20 &  0.84 &   18.81 & 20 & 0.81 & 19.08 & 10  & 10 & 0.87 & 18.14 & 10 & 10 & 0.91 & 28.38 \\
     \texttt{Map 1} &  10 & 24 &  0.93 &  77.04  & 24 & 0.88 & 71.85 & 10 & 16 & 0.94 & 72.15 & 10 &  16 & 0.95 & 158.71\\
    &  20 & 20 & 0.79  & 40.99   & 20 & 0.81 & 19.23  & 20 & 10 & 0.95 & 58.49 & 20 & 10 & 0.95 & 89.01 \\
    &  20 & 24 &  0.89 &  70.85  & 24 & 0.89  & 71.18 & 20 & 16 & 0.95 &609.22 & 20 & 16 & 0.95 & 734.45 \\
    \midrule
      {3D}      &  10 &  10 & 0.00 & 6.94 & 10 & 0.00 & 7.67 & 10 & 6 & 0.55 & 59.02 & 10 & 6 & 0.98 & 69.30 \\  
      \texttt{Aircraft}  &  20 &  20 & 0.00 & 78.27 & 20 & 0.00 & 81.45 & 20 & 6 & 0.82 & 263.37 & 20 & 6 & 0.99 & 284.34 \\  
           & 30 & 24 & 0.70 & 405.56  & 24 & 0.90 & 417.33 & 30 & 6 & 0.99 & 1123.64 & 30 & 6 & 0.99 & 1277.89 \\
           \midrule
      {2D Polynomial}      & 10 & 10 & 0.56 & 358.12 & 10 & 0.54 & 352.12 & 10 & 6 & 0.72 & 194.57 & 10 & 6 & 0.79 & 207.45   \\
      \texttt{Room Temp} & 10 & 12 & 0.57 & 415.62 & 12 & 0.58 & 588.78 & 10& 8 & 0.81 & 453.12 & 10 & 8 & 0.83 & 539.39  \\
       \midrule
      {3D Polynomial}  & 10 & 6 & 0.00 & 2644.91 & 6 & 0.00 & 3112.58 & 10 & 4 & 0.49 & 2214.56 & 10 & 4 & 0.53 & 2175.78 \\
      \texttt{Room Temp} & 10 & 8 & 0.05 & 5147.82 & 8 & 0.00 & 5289.36 & 10 & 6 & 0.74 & 5988.77 & 10 & 6 & 0.79 & 5991.17\\
      \bottomrule
\end{tabular}
}
\end{table*}

\subsubsection{Comparison to \cite{vzikelic2023learning}}
{First, to illustrate the effectiveness of our formulation, particularly the use of the relaxation variable $\alpha$, we conduct a comparative evaluation against the framework proposed in \cite{vzikelic2023learning}}.
As that approach is restricted to time-invariant, infinite-horizon certification, the benchmark is performed exclusively in the time-invariant setting. To ensure a direct and equitable comparison, both methods are formulated as SOS programs. Notably, the scaling imposed by the martingale condition in~\cite{vzikelic2023learning} reduces the flexibility of the resulting certificate. This structural restriction manifests in the empirical results reported in Table~\ref{tab:benchmark3}, which clearly demonstrate the advantage of our formulation.

Specifically, the results in Table~\ref{tab:benchmark3} show that for the 1D linear system, our time-invariant approach consistently produces significantly stronger certificates than~\cite{vzikelic2023learning}. At degree 8, our method achieves $\low{P}_\safe = 0.76$ compared to $0.16$, and at degree 10 it improves to $0.96$ while the baseline remains at $0.17$. The computational times are of the same order of magnitude for both approaches. 

For the 2D contraction benchmark, both methods show improvement as the polynomial degree increases; however, clear performance differences emerge as the initial set becomes more challenging. For $X_{\initial}^1$, our approach consistently produces substantially stronger certificates and succeeds in cases where the method of~\cite{vzikelic2023learning} fails. For $X_{\initial}^2$, the gap further widens, with our method steadily improving the certified lower bound as the degree increases, while the baseline approach yields noticeably weaker certificates. This trend is even more pronounced for $X_{\initial}^3$, where our approach continues to provide meaningful and progressively improving certificates, whereas the method of~\cite{vzikelic2023learning} delivers significantly weaker results and ultimately fails at higher degrees. The computational times are overall comparable across the two approaches.

\subsubsection{Comparison to \cite{xue2025finite}}

Next, to further illustrate the effectiveness of our formulation, we also compare against the framework in \cite{xue2025finite} using the results reported in their Example~2, as no implementation is publicly available. We replicate the same setup and evaluate both our time-invariant and time-varying formulations under an identical horizon of $H=50$. Since computation time is not reported in~\cite{xue2025finite}, we can only compare the probability bounds in Table~\ref{tab:benchmark2}.


\begin{table}[t]
\centering
\caption{Probability bound comparison against \cite{xue2025finite}. 
\texttt{TI} and \texttt{TV} are time-invariant and -varying, and
\texttt{deg} is the degree of SOS polynomial.
}
\resizebox{\linewidth}{!}{
\begin{tabular}{c|ccccccccc}
\toprule
\texttt{deg} & 2 & 4 & 6 & 8 & 10 & 12 & 14 & 16 & 18 \\
\midrule
\cite{xue2025finite}  & 0.16 & 0.28 & 0.35 & 0.37 & 0.46 & 0.52 & 0.57 & 0.58 & 0.61 \\
Our \texttt{TI} & \textbf{0.28}  & 0.28 & 0.37 & 0.38 & \textbf{0.71} & \textbf{0.73} & \textbf{0.73} & \textbf{0.73} & \textbf{0.75} \\
Our \texttt{TV} & \textbf{0.28} & \textbf{0.62} & \textbf{0.78} & \textbf{0.78} & - & - & - & - & - \\
\bottomrule
\end{tabular}
}
\label{tab:benchmark2}
\end{table}

As shown in Table~\ref{tab:benchmark2}, our time-invariant formulation consistently achieves stronger lower bounds on $P_\reach$ across all polynomial degrees. While both approaches improve with increasing degree, our method exhibits a marked gain at moderate degrees (e.g., $d \geq 10$), where the certified probability increases significantly compared to \cite{xue2025finite}, indicating reduced conservatism.
Moreover, the time-varying formulation provides substantial additional improvement, achieving significantly tighter bounds at much lower degrees. 


\subsubsection{Verification}

\begin{table*}[t]
\caption{Reach-avoid \texttt{synthesis} results for finite- and infinite-time horizons. \texttt{deg} is the degree of SOS polynomial, $\low{P}_\reach$ is the obtained lower-bound on $P_\reach$, and $\tau$ is the computation time in seconds.
}
\label{tab:reach-avoid-2}
\centering
\scalebox{1}
    {
\begin{tabular}{c|cccr | ccr | cccr | cccr}
\toprule 
 \toprule  
   & & \multicolumn{6}{c}{$\text{Time-Invariant Certificate}$}  
   & \multicolumn{8}{c}{\text{Time-Varying Certificate}}  \\
 \multirow{ 2}{*}{Model} & 
 & \multicolumn{3}{c}{$\texttt{Finite-Horizon}$}  
   & \multicolumn{3}{c}{\texttt{Infinite-Horizon}}  
    & \multicolumn{4}{c}{$\texttt{Finite-Horizon}$}  
   & \multicolumn{4}{c}{\texttt{Infinite-Horizon}}  \\
   \\
  & $H$ & \texttt{deg} &   $\low{P}_\reach$ & $\tau(s) \ $ 
   & \texttt{deg} & $\low{P}_\reach$   & $\tau(s) \ $ 
  & $H$ & \texttt{deg} &   $\low{P}_\reach$ & $\tau(s) \ \ $ 
   & \texttt{k} & \texttt{deg} & $\low{P}_\reach$   & $\tau(s)$ \ \ \\
 \midrule
  {1D Linear}   & 10 &  10 & 0.92 & 1.68 & 10 & 0.93 & 2.43 & 10 & 2 & 0.98 & 5.76 & 10 & 2  & 0.98 & 6.17 \\
  \eqref{eq:1D_linear_stable}   & 10 & 20 & 0.99  &3.56 & 20  & 0.99 &  3.78 & 10 & 4 & 0.99 & 10.72 & 10 & 4 & 0.99 & 11.01 \\ 
  & 20 & 10 & 0.93 & 1.89 & 10 & 0.93 & 2.58 & 20 & 2 & 0.99 & 17.84 & 20 & 2 & 0.99 & 18.18  \\ 
  &20 & 20 & 0.99 & 4.53 & 20 & 0.99 & 4.78 & 20 & 4 & 0.99 & 28.62 & 20 & 4 & 0.99 & 33.34\\ 
 \cline{2-16} 
    {1D Linear}  & 10& 10& 0.84& 1.15 & 10 & 0.83 & 2.49 & 10 & 2 & 0.97 & 5.14 & 10 & 2 & 0.98 & 5.99\\
 \eqref{eq:1D_linear_unstable}   & 10 & 20 & 0.98 & 3.89 & 20 & 0.98 & 3.45 & 10 & 4 & 0.99 & 10.45 & 10 & 4 & 0.99 & 11.59\\ 
  &20 & 10 &0.85 & 1.95& 10& 0.85& 2.97 & 20 & 2 & 0.99 & 18.56 & 20 & 4 & 0.99 & 18.99 \\
    &  20& 20& 0.99&5.89 & 20& 0.99& 5.78 & 20 & 4 & 0.99 & 29.13 & 20 & 4 & 0.99 & 31.33\\
\midrule
     {2D Contraction} & 10  & 20 & 0.88  & 32.19 &  20 &  0.88 & 32.70 & 10 & 4 & 0.92 & 59.45 & 10 & 4 & 0.93 & 63.63\\
     \texttt{Map 1} & 10 & 24 & 0.99  & 98.41 & 24 & 0.99 & 98.43 & 10 & 6 & 0.99 & 108.88 & 10 & 6 & 0.99 & 117.14 \\
     & 20 & 20 & 0.88 & 30.19  &  20 & 0.88 & 35.91 & 20 &  4 & 0.99 & 388.24 & 20 & 4 & 0.99 & 400.14 \\
    & 20 & 24& 0.99 &  101.18 & 24 & 0.99 & 108.25 & 20 & 6 & 0.99 & 924.34 & 20 & 6 & 0.99 & 968.42 \\
\cline{2-16} 
    {2D Contraction}  & 10 & 20 & 0.76 & 58.72 & 20 & 0.76 & 88.72 & 10 & 4 & 0.85 & 68.15 & 10 & 4 & 0.86 & 70.06 \\
 \texttt{Map 2}   & 10 & 24 & 0.77 & 123.23 & 20 & 0.77 & 99.17 & 10 & 6 & 0.91 & 115.83 & 10 & 6 & 0.91 & 142.18 \\ 
   & 20 & 20 & 0.76 & 68.45& 20 & 0.76 & 77.52 & 20 & 4 & 0.92 & 426.27 & 20 & 4 & 0.92 & 455.89 \\
    & 20 & 24 & 0.77 & 115.68 & 20 & 0.77 & 119.79 & 20 & 6 & 0.93 & 1075.69 & 20 & 6 & 0.93 & 1089.23\\
    \midrule
{3D} & 10 & 4 & 0.00 & 714.34 & 4 & 0.00 & 698.67 & 10 & 2 & 0.87 & 472.43 & 10 & 2 & 0.89 & 485.31\\
\texttt{Aircraft} & 10 & 6 & 0.00 & 1318.51 & 6 & 0.00 & 1367.48 & 10 & 4 & 0.90 & 1599.23 & 10 & 4 & 0.91 & 1685.34\\
  \midrule
    {2D Polynomial}& 10 & 10 & 0.82 & 374.87 & 10 &0.81  & 399.45 & 10 & 4 & 0.85 & 1235.67 & 10 & 4 & 0.85 & 1344.22 \\
 \texttt{Room Temp} & 10 & 12 & 0.89 & 592.13 & 10 & 0.88 & 623.48 & 10 & 6 & 0.91 & 2744.10 & 10 & 6 & 0.92 & 2800.05\\ 
 \midrule
    {3D Polynomial}  & 10 & 4 & 0.00 & 855.14 & 4 & 0.00 & 928.59 & 10 & 2 & 0.39 & 511.20 & 10 & 2 & 0.40 & 577.43 \\
     \texttt{Room Temp} & 10 & 6 &  0.30 & 7855.12 & 6 & 0.31  & 7422.36 & 10 & 4 & 0.75 & 5127.89 & 10 & 4 & 0.76 & 5344.77\\
    & 10 & 8 & 0.42 & 8454.23 & 8 & 0.41 & 8454.57 & 10 &  6  & 0.99 & 14013.56 & 10 & 6 & 0.99 & 14083.65\\ 
 \bottomrule
\end{tabular}
}
\end{table*}

Table~\ref{tab:reach-avoid-1} summarizes the reach-avoid verification results using time-invariant and time-varying certificates.
For the 1D systems, the SOS-based lower-bound reach probabilities closely match empirical Monte Carlo estimates, with higher polynomial degrees improving the bounds while keeping computation times modest. In 2D, the contraction map exhibits more challenging dynamics due to a stronger rotational component, yet the SOS certificates achieve values consistent with Monte Carlo, demonstrating their ability to capture the stochastic behavior accurately. 

For 3D linear and polynomial systems, the verification problem becomes more demanding, yet especially time-varying certificates still provide rigorous lower bounds that capture nontrivial reach-avoid behavior.  Note that in the case of 3D aircraft, a certificate of degree 24 is needed to get a nontrivial probability bound.  However, with time-varying certificate, even degree 6 produces very high lower bound of $0.98$ at about $70s$ compute time.

In fact, time-varying certificates consistently yield higher lower-bound probabilities compared to single-certificate formulations, often achieving similar or better bounds with lower-degree polynomials. This improvement generally comes at the cost of increased computational effort, but it is demonstrated that time-varying certificates are a practical approach for obtaining tighter guarantees without excessively increasing polynomial degrees.

A representative certificate for the 2D contraction map is shown in Fig.~\ref{fig:certificate}, illustrating high values within the reach and initial sets and a sharp decay toward the unsafe region. 
Visualization of the
time-varying certificates for the 1D Linear System~\eqref{eq:1D_linear_unstable}, are provided in Appendix~\ref{app:cert}. These results highlight that SOS-based verification can provide rigorous finite and infinite-time horizon \emph{reach-avoid} guarantees.

\subsubsection{Control Synthesis}
{Finally}, we consider the reach-avoid joint feedback control and certificate synthesis problem for the same systems. Table~\ref{tab:reach-avoid-2} reports the lower-bound reach-avoid probabilities obtained when designing \emph{linear feedback} controllers over finite- and infinite-time horizons, for both time-invariant and time-varying formulations. Fig.~\ref{fig:mc-synthesis} shows Monte Carlo simulations.

For the 1D linear systems, the synthesized controllers significantly increase the reach probability (especially for time-invariant certificates), bringing the SOS-based guarantees very close to the newly simulated empirical Monte Carlo probabilities. Time-varying certificates further improve the lower bounds, often achieving similar or better performance with very lower-degree polynomials (e.g., degree 2), although at the cost of longer computation times.

\begin{figure}[t]
    \centering
    \includegraphics[width=0.7\columnwidth]{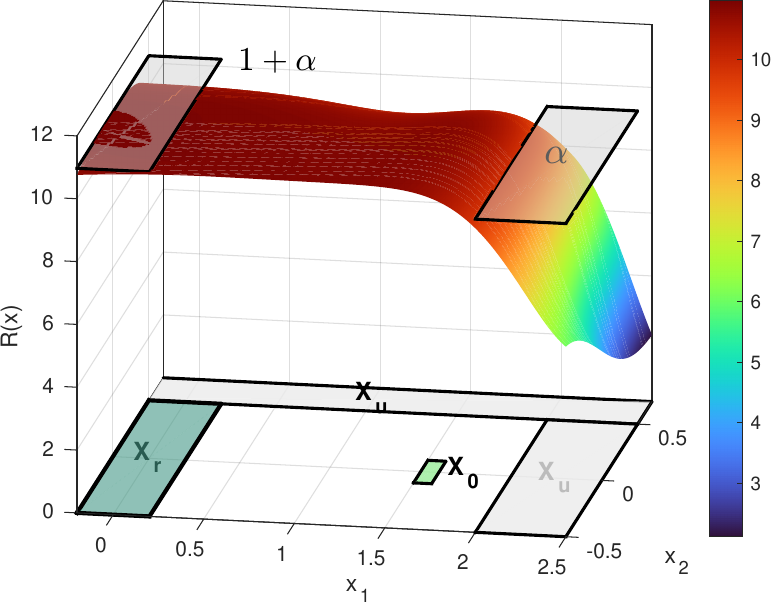}
    \caption{Reach-avoid certificate for Contraction Map 1. 
    }
    \label{fig:certificate}
\end{figure}

In 2D, for Contraction Map 1, the controller improve $P_r$.  Note that for both cases of time-invariant certificate (both finite- and infinite-horizon), degree 24 polynomial is needed to achieve $P_r \geq 0.99$, whereas the time-varying certificate can achieve that with degree $6$ at the additional cost of $10s$ computation time.

Contraction Map 2 substantially benefits from linear feedback controller, nearly matching the Monte Carlo estimate, that increased from $P_\reach \approx 0.19$ to $P_\reach \approx 0.95$.  In this case, the time-invariant certificate, even with degree 24, performs poorly ($P_s \geq 0.77)$, whereas the time-varying certificate can reach to $0.93$ with degree 6 but at the cost of 18 min computation time.

Across the 3D aircraft and room-temperature systems, time-varying certificates consistently yield higher lower-bound probabilities than single-certificate synthesis, demonstrating the advantage of time-varying formulations in higher-dimensional settings.
Overall, these results show that linear feedback control can meaningfully enhance reach probabilities relative to open-loop dynamics. Figure~\ref{fig:mc-synthesis} presents the empirical Monte Carlo reach probabilities for the synthesized controllers across several benchmark systems, as well rolled-out trajectories under the policies, illustrating how the controllers concentrate trajectories within the reach and initial sets, while reducing excursions toward unsafe states.

\subsubsection{Overall Observations}
These empirical evaluations show that time-varying certificates are more effective for high-dimensional systems. This is because time-invariant certificates typically require higher-degree polynomials to produce nontrivial probabilistic lower bounds. However, the number of decision variables in a polynomial grows exponentially with both the polynomial degree and the state dimension, making the synthesis of high-degree polynomials in large-dimensional spaces computationally prohibitive.

In contrast, time-varying certificates can achieve nontrivial probabilistic lower bounds using low-degree polynomials, albeit over a larger number of time steps. Since increasing the number of time steps scales the certificate parameters only linearly, time-varying certificates offer substantially better scalability and performance than time-invariant ones in high-dimensional settings.

\section{Conclusion}

This paper presents a reach-avoid certificate framework for discrete-time, continuous-space stochastic systems over both finite- and infinite-horizon settings. We formulate reach-avoid analysis as a convex sum-of-squares optimization problem, enabling both verification of given controllers and joint synthesis of optimal feedback controllers with corresponding time-invariant and time-varying certificates. The proposed approach provides theoretical guarantees for safety and goal-reaching performance, and its effectiveness is demonstrated across a range of linear and nonlinear stochastic systems through benchmark case studies. Overall, this framework offers a practical and rigorous tool for certifying and controlling stochastic systems in continuous state and action spaces.

\appendix
\subsection{Certificate Visualization}
\label{app:cert}


The time-varying certificates of the unstable linear system is visualized in Fig.~\ref{fig:certificates}. The first snapshot exhibits values below the corresponding $\alpha$ throughout both the safe and unsafe sets, reflecting the intended initialization of the certificate. By the final horizon, the trend aligns with the previous description: the certificate is high in the reach set, low in the unsafe set, and emphasizes the initial set where the corresponding $\gamma$ is maximized. These plots illustrate that both time-invariant and time-varying certificates effectively capture the reach-avoid optimization structure.


\begin{figure}[h!]
    \centering
    \includegraphics[width=0.15\textwidth]{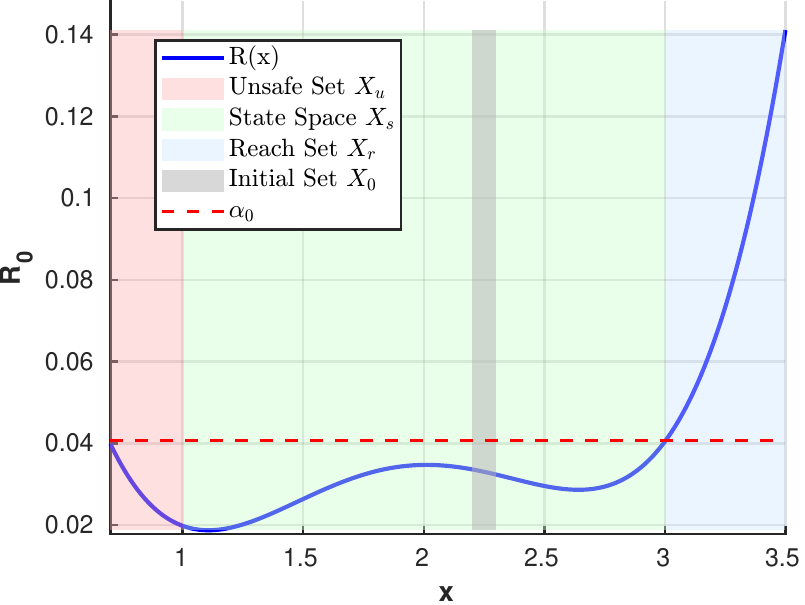}\hfill
    \includegraphics[width=0.15\textwidth]{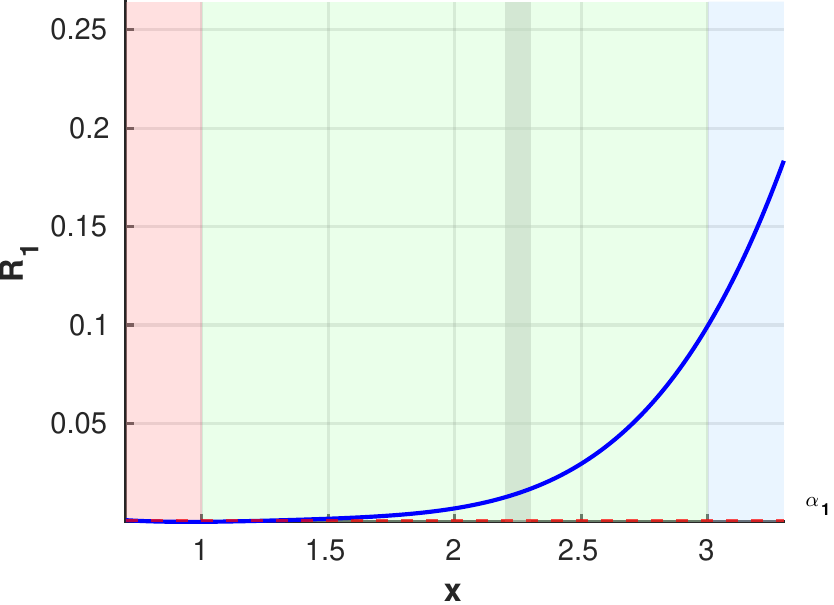}\hfill
    \includegraphics[width=0.15\textwidth]{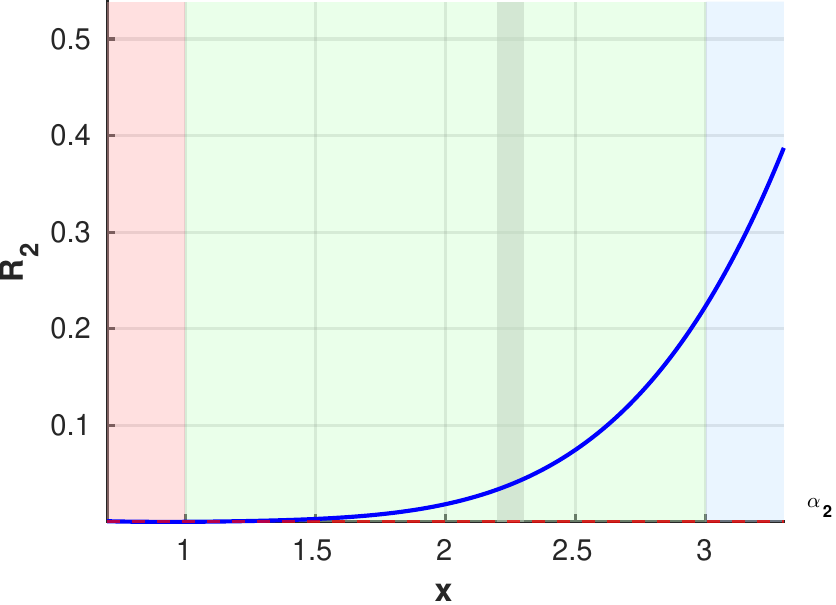}\\[3pt]

    \includegraphics[width=0.15\textwidth]{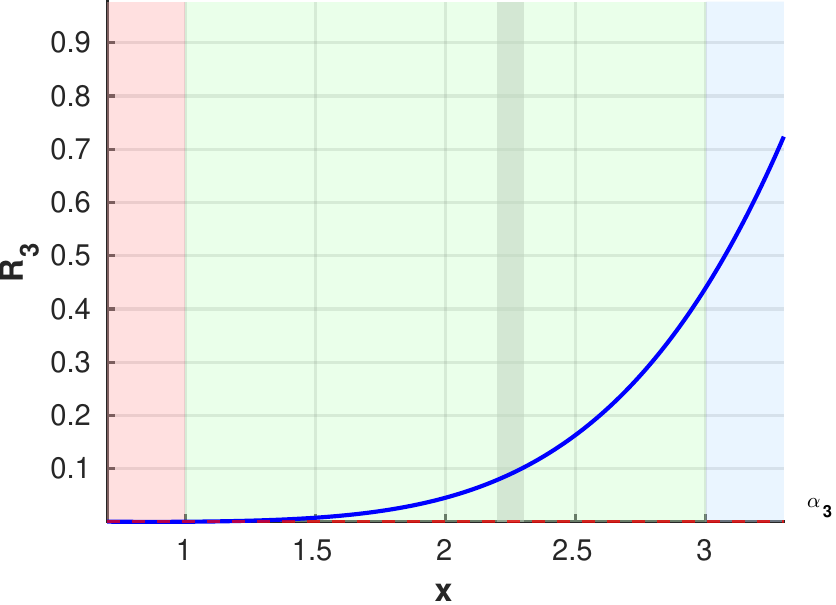}\hfill
    \includegraphics[width=0.15\textwidth]{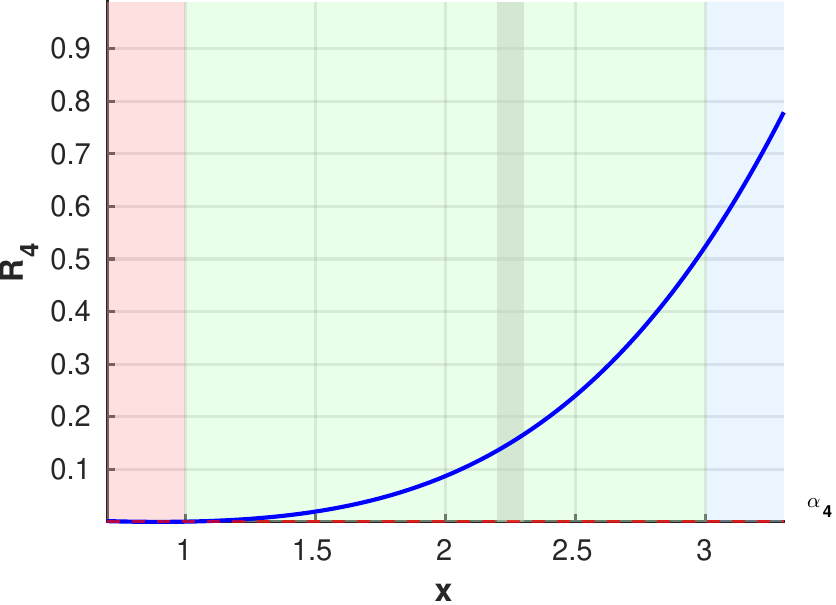}\hfill
    \includegraphics[width=0.15\textwidth]{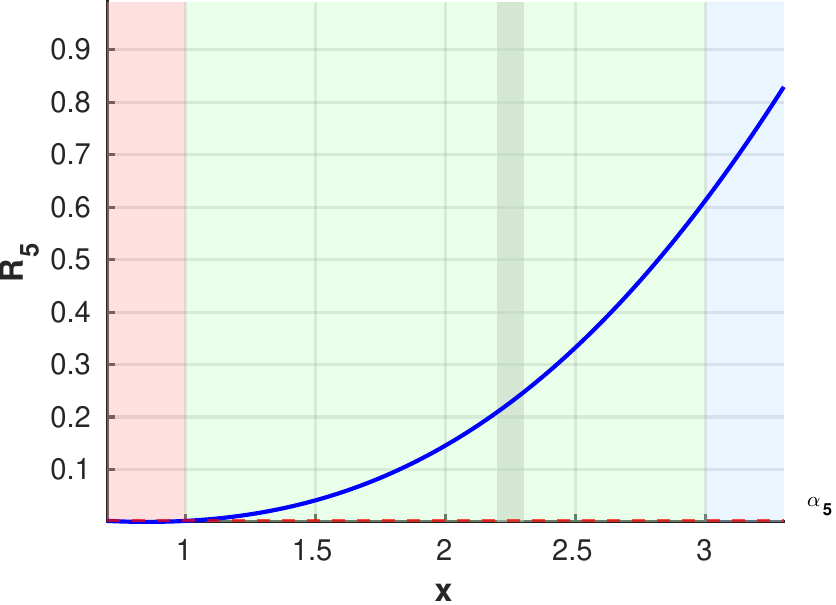}\\[3pt]

    \includegraphics[width=0.15\textwidth]{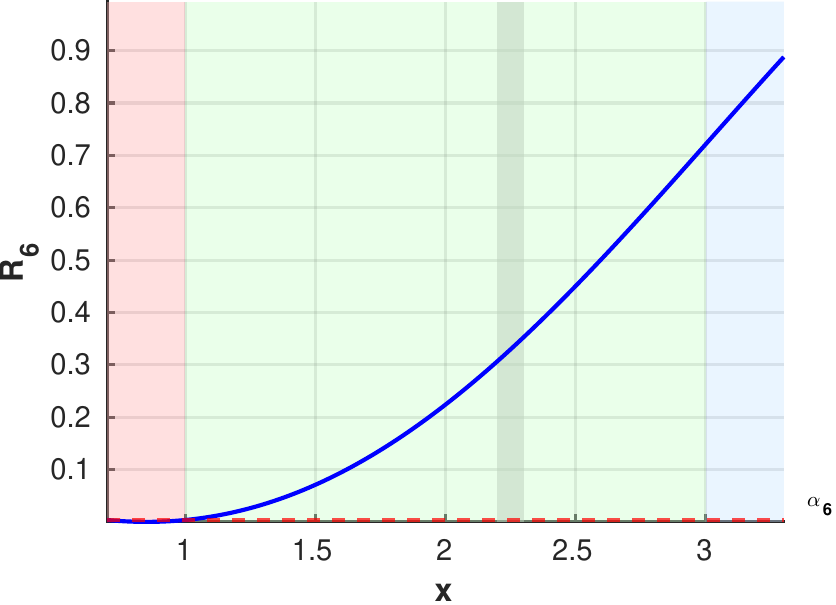}\hfill
    \includegraphics[width=0.15\textwidth]{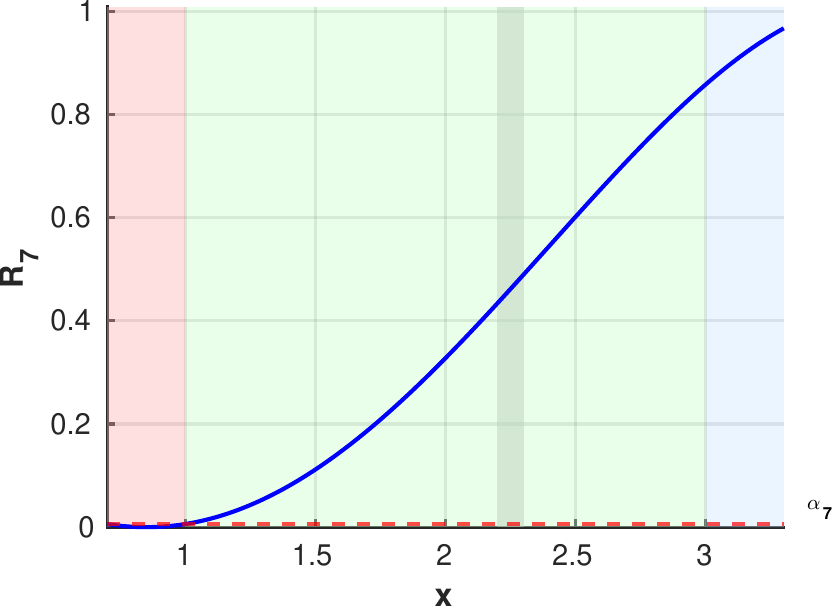}\hfill
    \includegraphics[width=0.15\textwidth]{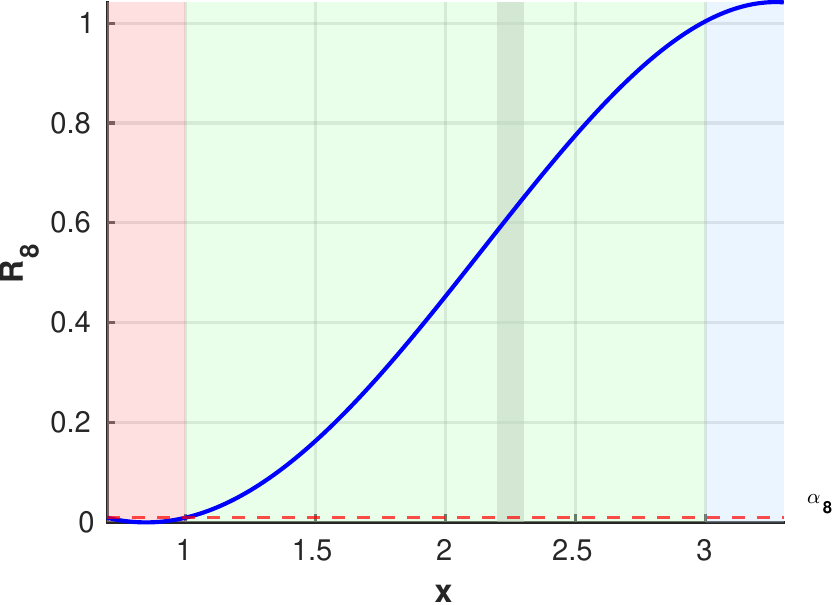}\\[3pt]

    \includegraphics[width=0.15\textwidth]{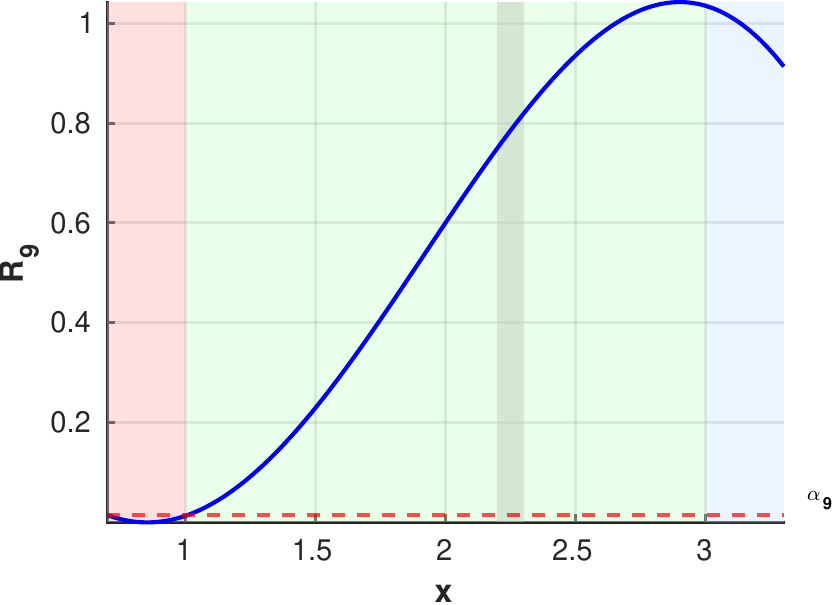}\hfill
    \includegraphics[width=0.15\textwidth]{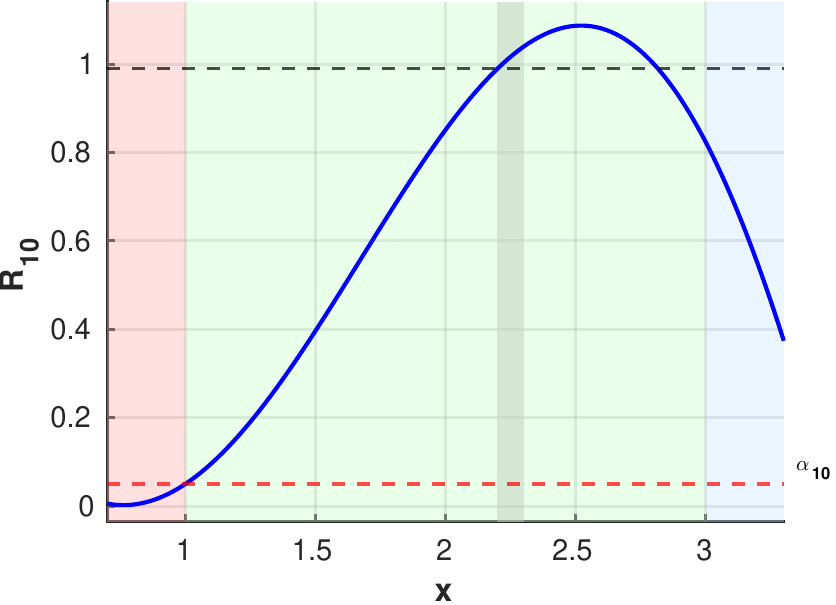}

    \caption{Time-varying certificate plots for horizon $H=10$. Certificate evolution over time: initially below $\alpha_0$ in both safe and unsafe sets, then by horizon $H$, low values in the unsafe set, and high near the initial set where $\gamma$ is maximized.}
    \label{fig:certificates}
\end{figure}

\subsection{Aircraft Model }
\label{parameters}

For the aircraft model, the following parameters are used:  
Mass: $m = 1.9\,\mathrm{kg}$; \quad 
Pitch-axis inertia: $I_{yy} = 0.025\,\mathrm{kg\,m^2}$; \quad
Trim airspeed: $U_0 = 13.5\,\mathrm{m/s}$; \quad
Trim pitch angle: $\theta_0 = 0$; \quad
Gravity: $g = 9.81\,\mathrm{m/s^2}$.
    
Longitudinal derivatives:
$X_u = -0.30\,\mathrm{kg/s}$; \quad
$X_w = 2.8\,\mathrm{kg/s}$; \quad
$Z_u = -5.2\,\mathrm{kg/s}$; \quad
$Z_w = -15.0\,\mathrm{kg/s}$; \quad
$Z_{\dot w} = -1.8\,\mathrm{kg}$; \quad
$Z_q = -3.9\,\mathrm{kg\,m/s}$; \quad
$M_u = -0.0025\,\mathrm{N\,m\,s/m}$; \quad
$M_w = -0.08\,\mathrm{N\,m\,s/m}$; \quad
$M_q = -0.50\,\mathrm{N\,m\,s}$; \quad

Further, the aircraft operates on the safe set $X_\safe = [1, 3]^3$, starting from initial region
$X_0 = [1.5, 1.6] \times [2.5, 2.6]] \times [1.5, 1.6]$. The reach set is located at $X_\reach = [2.7, 3.0] \times [1.0, 3.0] \times [1.0, 3.0]$.


\section*{References}
\bibliographystyle{IEEEtran}
\bibliography{TAC/cite}

\begin{IEEEbiography}[{\includegraphics[width=1in,height=1.25in,clip,keepaspectratio]{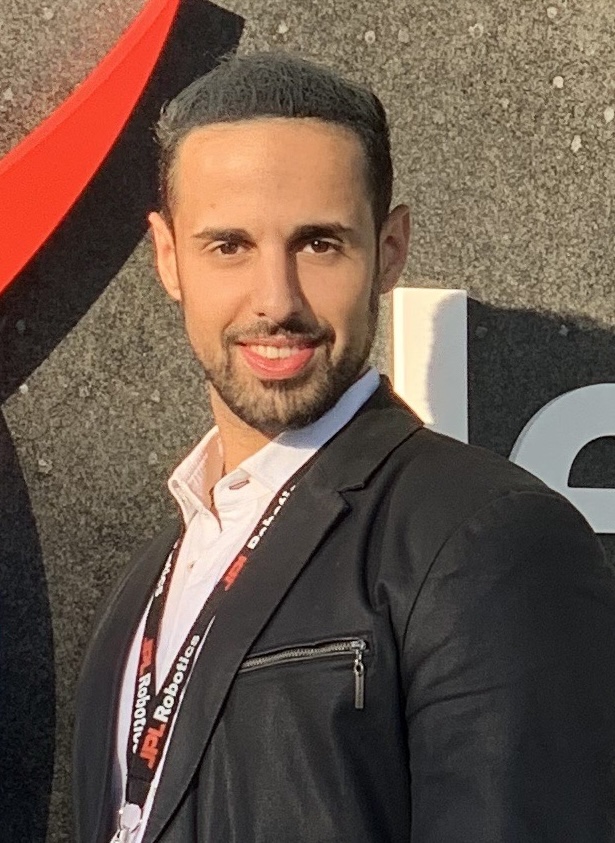}}]{Rayan Mazouz} (Member, IEEE)
received his B.Sc. and M.S.c degrees in Aerospace Engineering from Delft University of Technology. Currently, he is persuing a Ph.D. at the University of Colorado in Boulder. His research primarily focuses on the formal verification and control synthesis of stochastic dynamical systems. 
\end{IEEEbiography}

\begin{IEEEbiography}[{\includegraphics[width=1in,height=1.25in,clip,keepaspectratio]{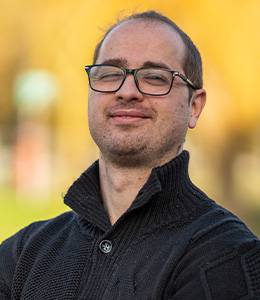}}]{Luca Laurenti} (Member, IEEE) is a tenure track assistant professor at the Delft Center for Systems and Control at TU Delft and head of the Research on Robust and Intelligent Autonomous Systems Lab at AI4I. He received his PhD from the Department of Computer Science, University of Oxford (UK), where he was a member of the Trinity College. Luca has a background in stochastic systems, control theory, formal methods, and artificial intelligence. His research work focuses on developing data-driven systems that are provably robust to interactions with a dynamic and uncertain world.
\end{IEEEbiography}

\begin{IEEEbiography}[{\includegraphics[width=1in,height=1.25in,clip,keepaspectratio]{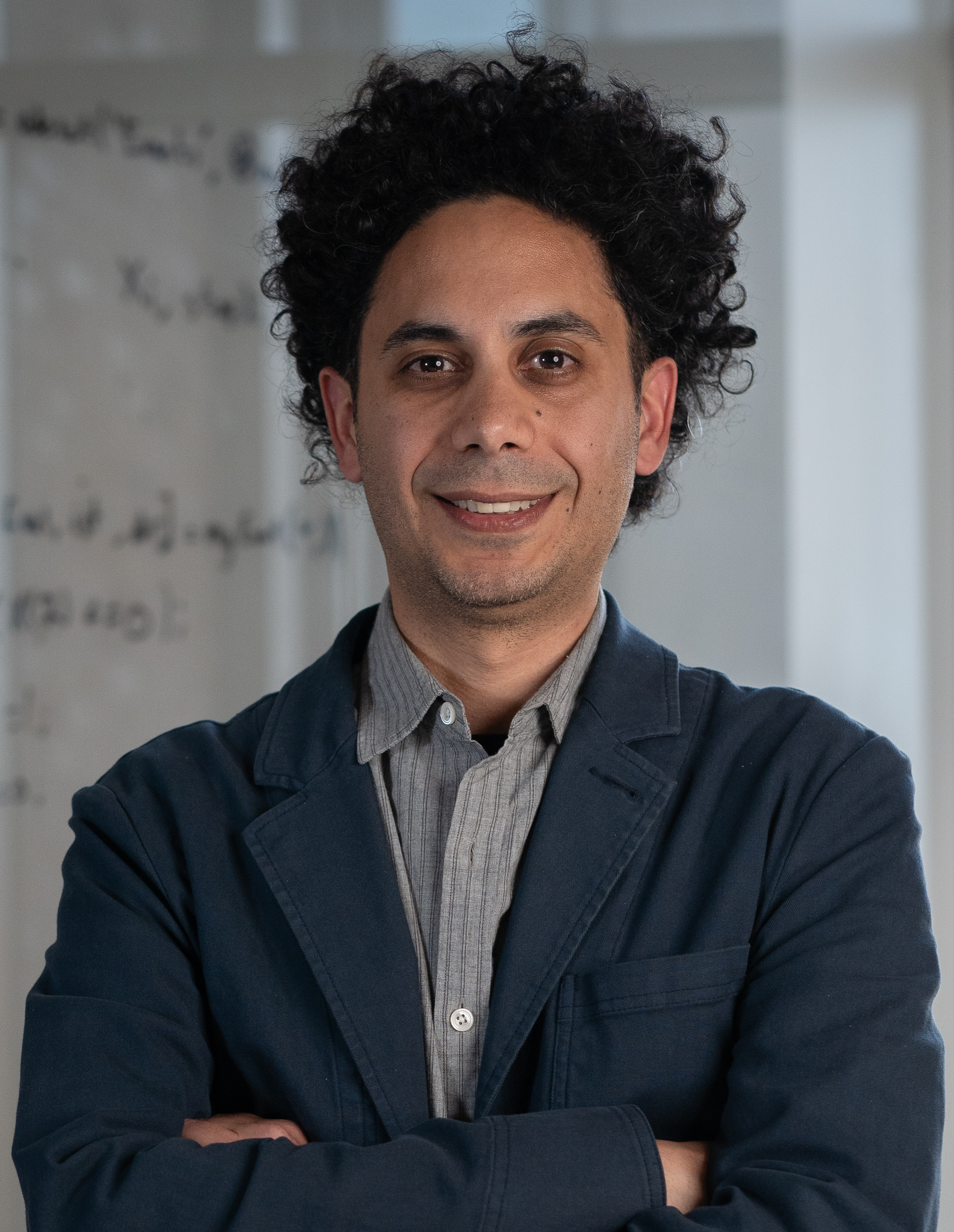}}]{Morteza Lahijanian} (Member, IEEE)
is an Associate Professor in the Aerospace Engineering Sciences department, an affiliated faculty at the Computer Science department and Robotics Program, and the director of the Assured, Reliable, and Interactive Autonomous (ARIA) Systems group at the University of Colorado Boulder. 
His awards include Ella Mae Lawrence R. Quarles Physical Science Achievement Award, Jack White Engineering Physics Award, outstanding junior faculty award, and NSF GK-12 Fellowship. 
Dr. Lahijanian's research interests span the areas of control theory, stochastic hybrid systems, formal methods, machine learning, and game theory with applications in robotics, particularly, motion planning, strategy synthesis, model checking, and human-robot interaction. 
\end{IEEEbiography}

\end{document}